\documentclass[journal,twoside,web]{ieeecolor}

\usepackage{generic}
\def\BibTeX{{\rm B\kern-.05em{\sc i\kern-.025em b}\kern-.08em
		T\kern-.1667em\lower.7ex\hbox{E}\kern-.125emX}}
\usepackage{latexsym, epsfig, amsthm, amssymb, amsmath}
\usepackage{showidx}
\usepackage{tikz}

\usepackage{tabularx, floatrow, ragged2e,booktabs} 
\usepackage{latexsym, epsfig, tikz, subfigure}
\usepackage{verbatim}
\usepgflibrary{arrows}
\usepackage{verbdef}
\usepackage{color}
\usepackage{graphicx}
\usepackage[european]{circuitikz}
\usepackage{arydshln}
\usepackage[T1]{fontenc}
\usepackage[thinlines]{easytable}

\usepackage{mathtools}
\usepackage[amsmath]{empheq}

\usepackage{caption}
\tikzset{>=latex}

\tikzset{%
block/.style    = {draw, thick, rectangle, minimum height = 3em,
	minimum width = 3em},
block1/.style    = {draw, thick, rectangle, minimum height = 1.7em,
	minimum width = 1.7em,fill=gray!70},
block2/.style    = {draw, thick, rectangle, minimum height = 1.7em,
	minimum width = 1.7em,fill=gray!30},
sum/.style      = {draw, circle, node distance = 1.8cm}, 
}

\title{Network Realization Functions for \\  Optimal Distributed Control\vspace{-2mm}}

\author{
\c{S}erban Sab\u{a}u$^\sharp$,\thanks{$\sharp$ \c{S}erban Sab\u{a}u
	is with the Electrical and Computer Engineering Department, Stevens Institute of Technology, Hoboken, NJ 07030 USA (e-mail:
	ssabau@stevens.edu).} Andrei {S}peril\u{a}$^\ddag$\thanks{$\ddag$ Andrei {S}peril\u{a} and Cristian
	Oar\u{a} are with the Faculty of Automatic Control and Computers,
	``Politehnica" University of Bucharest, Sector 6, 060042 Romania (e-mails: \{andrei.sperila, cristian.oara\}@upb.ro). }, Cristian Oar\u{a}$^\ddag$ 
and Ali Jadbabaie$^\flat$ \thanks{$\flat$ Ali Jadbabaie is with the Institute for Data, Systems and Society, Massachusetts Institute of Technology (MIT), Cambridge, MA 02139 USA (email: jadbabai@mit.edu).}
	\thanks{The work of \c Serban Sab\u au was supported by the NSF - CAREER under Grant 1653756. The work of Andrei Speril\u a and Cristian Oar\u a was supported by a grant of the Ministry of Research, Innovation and Digitization, CCCDI - UEFISCDI, project no. PN-III-P2-2.1-PED-2021-1626, within PNCDI III. The work of Ali Jadbabaie was supported by the Vannevar Bush Fellowship from the Office of the Secretary of Defense.}
\vspace{-10mm}
}

\newcommand{\norm}[1]{\lVert#1\rVert}
\renewcommand{\tilde}{\widetilde}

\newcommand{\FF}{{{\rm I \kern -0.2em R}}}
\newcommand{\RR}{{{\rm I \kern -0.2em R}}}

\newcommand{\CC}{{{\mbox{\rm \hspace*{0.05ex}
			\rule[.18ex]{.18ex}{1.24ex} \kern -.65em C}}}} 

\newcommand{\bea}{\begin{eqnarray}}
\newcommand{\eea}{\end{eqnarray}}

\newtheorem{theorem}{Theorem}[section]
\newtheorem{rem}[theorem]{Remark}

\newtheorem{lem}[theorem]{Lemma}
\newtheorem{coro}[theorem]{Corollary}
\newtheorem{defn}[theorem]{Definition}
\newtheorem{assumption}[theorem]{Assumption}

\newcommand{\ba}{\left[ \begin{array}}
\newcommand{\baa}{\begin{array}}
	\newcommand{\ea}{\end{array} \right]}
\newcommand{\eaa}{\end{array}}

\newcommand{\be}{\begin{equation}}
\newcommand{\ee}{\end{equation}}

\newcommand{\la}{\lambda} 

\renewcommand{\tt}{\times}

\def\math#1{\ifmmode{#1} \else {$#1$}\fi}

\newcommand{\sg}{\ifmmode \Sigma \else $\Sigma$ \fi}

\date{}

\begin{document}
	\maketitle
	\thispagestyle{empty}
	
	\begin{abstract} 
		In this paper, we discuss a distributed control architecture, aimed at networks with linear and time-invariant dynamics, which is amenable to convex formulations for controller design. The proposed approach is well suited for large scale systems, since the resulting feedback schemes completely avoid the exchange of internal states, i.e., plant or controller states, among sub-controllers. Additionally, we provide state-space formulas for these sub-controllers, able to be implemented in a distributed manner.\vspace{-4mm}
	\end{abstract}
	
	\begin{IEEEkeywords}
		Distributed control, linear time-invariant networks, scalable implementations.\vspace{-3mm} 
	\end{IEEEkeywords}
	
	\section{Introduction} \label{Introdc}\vspace{-2mm}
	
	\subsection{Scope of Work}\vspace{-1mm}
	
	The multi-faceted intricacies of the optimal decentralized control problem are widely recognized in literature. With the hope for the existence of any convenient (let alone convex) parameterizations dispelled (see, for example, \cite{DateChow}), recent research advances have resorted to modern convexification or regularization methods, such as \cite{Lavaei,Jovanovic, reg, Sznaier, VoulgarisStudent}. In this context, the so-called {\em System Level Synthesis} ({SLS}) methods from \cite{matni1,matni2} provided an insightful perspective on distributed controller design, by exploiting the classical work from \cite{BBS}. The connections between SLS and classical parameterizations of stabilizing controllers have been further elaborated upon in \cite{Luca1,Luca3} and, more recently, in \cite{Tseng} (see also the literature review from the introduction of our companion paper \cite{aug_sparse}).
	
	However, the SLS framework: \emph{(a)} necessitates implementations which communicate internal states, \emph{i.e.}, controller or plant states, thus producing Transfer Function Matrices ({TFM}) of the controller's representation with dimensions equal to that of the plant's state vector
	, while \emph{(b)} allowing for the direct application (see the concluding remark of Section II-B in \cite{JA}) of the scalable, \emph{specialized implementations} from Section III-C of \cite{matni2} only for networks which have open-loop stable plants.\vspace{-3mm}  
	
	\subsection{Problem Statement and Contributions}\vspace{-1mm}
	
	The paper tackles the problem of obtaining a set of stabilizing and distributed control laws of the type
	$$
		u={\bf \Phi} u+ {\bf \Gamma} z,
	$$
	where the pair of LTI (Linear Time Invariant) filters $(\mathbf{\Phi},\mathbf{\Gamma})$ showcase pre-specified sparsity patterns and are used to compute the command signals, $u$, via the regulated measurements, $z$, without communicating any internal states between sub-controllers. The main contributions of our work are as follows:
	\begin{enumerate}
		\item providing a general method of obtaining distributed control laws that are akin to the \emph{specialized implementations} presented in Section III-C of \cite{matni2};
		
		\item bypassing the drawbacks mentioned in points \emph{(a)} and \emph{(b)}, while guaranteeing the full scalability of the distributed control laws, for possibly unstable plants;
		
		\item obtaining control laws for both discrete- and continuous-time (see \cite{ aug_sparse,TAC2016}) systems, while completely avoiding any self-loops (either integrators or delay elements) on the control signals, thus facilitating implementation.
	\end{enumerate}
	We point out that the proposed method relies on the concept of Network Realization Functions ({NRF}), heavily inspired by the work in \cite{DSFARX} and \cite{DSF}, which is able to impose sparsity patterns directly on the distributed controller's coprime factors. The close affinity between these factorizations and NRFs (see \cite{SRTR}) enables the exploitation of the robust stabilization machinery for distributed controller design, as established in \cite{aug_sparse}.
	
	
	\vspace{-4mm}
	\subsection{Paper Structure}
	
	In Section \ref{sec:prelim}, we present the concept of NRF pairs, all while showcasing their defining traits. In Section \ref{sec:NRF}, we discuss the means of enforcing sparsity patterns upon an NRF pair, in the same vein as the approach from \cite{SI}, and we offer guarantees of closed-loop stability when implementing controllers via this formalism. We also present the means to obtain distributed state-space implementations for our controllers' NRF pairs, in contrast to the purely TFM-based perspective from \cite{SI}, and we show that norm-based optimal design in the NRF paradigm reduces to an affine model matching problem (see also \cite{aug_sparse}). In Section \ref{sec:alt}, we discuss alternative representations for our distributed control laws, akin to and inspired by the ones in \cite{matni1} and \cite{matni2}, which use closed-loop maps to compute the command signals. We show that, unlike the architecture from Section \ref{sec:NRF}, these cannot ensure closed-loop stability if the plant does not satisfy certain stability assumptions (as in Section III-C of \cite{matni2}). Section \ref{sec:ex} showcases a numerical example (see also Section V of \cite{aug_sparse}), while Section \ref{sec:fin} contains a number of conclusions.
	
	\vspace{-3mm}
	
	\section{General Setup and Technical Preliminaries}\label{sec:prelim}\vspace{-1mm}
	
	\subsection{Notation}
	
	Since the enclosed results are valid for both continuous- and discrete-time LTI systems, we denote by $\la$ the complex variable associated with the Laplace transform for continuous-time systems, or with the $\mathcal{Z}$-transform for discrete-time ones.
	
	We denote by $\mathbb{R}$ the set of real numbers and by $\mathbb{N}$ the set of natural ones, while $\mathbb{C}$ stands for the complex plane. Let $\mathbb{R}^{p\times m}$ be the set of $p \times m$ real matrices and $\mathbb{R}(\lambda)^{p\times m}$ the set of  $p \times m$ TFMs, matrices with real-rational functions as entries. Let $e_i$  stand for the $i^{\text th}$ vector in the canonical basis of $\mathbb{R}^{m\times 1}$. A TFM for which $\lim_{\lambda\rightarrow\infty} \mathbf{G}(\lambda)$ has only finite entries is called proper and it is called strictly proper when $\lim_{\lambda\rightarrow\infty} \mathbf{G}(\lambda)=O$. We denote by $\mathbb{R}_p(\lambda)^{p\times m}$ the set of proper TFMs. Finally, let $\mathbb{S}$ denote the domain of stability, which is either the open left-half plane ($\text{Re}(\lambda)<0,\ \lambda\in\mathbb{C}$) for continuous-time systems or the open unit disk ($|\lambda|<1,\ \lambda\in\mathbb{C}$) for discrete-time ones. 

	\subsection{TFM and Realization Theory}

	The LTI systems considered in this paper are described in the time domain by the classical state-space equations
	\begin{subequations}
		\begin{align} 
			\sigma{ x}  &=  A x  +  Bu  , & \label{ss0ab} \\ 
			y &=  C  x  + D u, & \label{ss0c}
		\end{align}
	\end{subequations}
	\noindent  where  $A\in\mathbb{R}^{n\times n}$, $B\in\mathbb{R}^{n\times m}$, $C\in\mathbb{R}^{p\times n}$ and $D\in\mathbb{R}^{p\times m}$, and $n$ is called the {\em order} of the realization. The time-domain operator denoted by $\sigma$ in \eqref{ss0ab}-\eqref{ss0c} stands for either the time derivative (in the continuous-time context) or the forward unit-shift (in the discrete-time case). For any $n$-dimensional state-space realization \eqref{ss0ab}-\eqref{ss0c}, the system's TFM is given by
	\begin{equation*}
		\mathbf G(\lambda) =  \left[\begin{array}{c|c}A-\lambda I_n & B \\ \hline C & D \end{array}\right] := D + C(\lambda I_n - A)^{-1}B.
	\end{equation*}


	
	We denote by $\mathcal{P}_u(\mathbf{G})$ the collection of (finite) unstable poles (see Section 6.5.3 of \cite{Kai}), \emph{i.e.}, located in $\mathbb{C}\backslash\mathbb{S}$, which belong to $\mathbf{G}\in\mathbb{R}_p(\la)^{p \times m}$. We refer to $\mathbf{G}\in\mathbb{R}_p(\la)^{p \times m}$ as {\em stable} if $\mathcal{P}_u(\mathbf{G})=\{\emptyset\}$ and we denote the set of these TFMs by $\mathbb{R}_{\mathbb{S}}(\lambda)^{p\times m}(\subset\mathbb{R}_p(\lambda)^{p\times m})$. Otherwise, we say that it is \emph{unstable}. Note that $\mathcal{P}_u(\mathbf{G})$ includes repeated terms when the unstable poles of $\mathbf{G}$ have multiplicities greater than $1$. Furthermore, we denote by $\Lambda_u(A)$ the collection of eigenvalues belonging to the matrix $A$ which are unstable, \emph{i.e.}, located in $\mathbb{C}\backslash\mathbb{S}$. Once again, $\Lambda_u(A)$ includes repeated terms when the unstable eigenvalues of $A$ have multiplicities greater than $1$.

		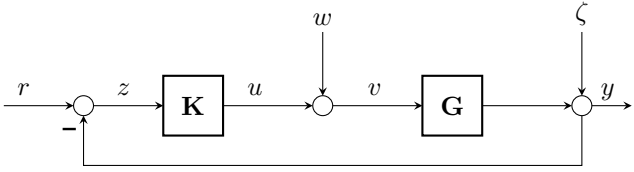
\begin{figure}[t]
			\begin{tikzpicture}[scale=0.265]
				\draw[xshift=0.1cm, >=stealth ] [->] (0,0) -- (3.5,0);
				\draw[ xshift=0.1cm ]  (4,0) circle(0.5);
				\draw[xshift=0.1cm] (3,1)   node {\bf{ }} (1,0.8) node {$r$};
				\draw [xshift=0.1cm](6,0.8)   node {$z$} ;
				\draw[ xshift=0.1cm,  >=stealth] [->] (4.5,0) -- (8,0);
				\draw[ thick, xshift=0.1cm]  (8,-1.5) rectangle +(3,3);
				\draw [xshift=0.1cm](9.5,0)   node {{${\bf K}$}} ;
				\draw[ xshift=0.1cm,  >=stealth] [->] (11,0) -- (15.5,0);
				\draw[ xshift=0.1cm ]  (16,0) circle(0.5cm);
				\draw [xshift=0.1cm](12.6,0.8)   node {$u$} ;
				\draw [xshift=0.1cm](18.6,0.8)   node {$v$} ;
				\draw [xshift=0.1cm] (14.8,0.7)   node {\bf{ }};
				\draw[  xshift=0.1cm,  >=stealth] [->] (16,3.7) -- (16,0.5);
				\draw [xshift=0.1cm] (16,3.6)  node[anchor=south] {$w$}  (15.3,1.5)  node {\bf{ }};
				\draw[  xshift=0.1cm,  >=stealth] [->] (16.5,0) -- (21,0);
				\draw[ thick, xshift=0.1cm ]  (21,-1.5) rectangle +(3,3) ;
				\draw [xshift=0.1cm] (22.5,0)   node {{${\bf G}$}} ;
				\draw[ xshift=0.1cm,  >=stealth] [->] (24,0) -- (28.5,0);
				\draw[ xshift=0.1cm ] (29,0)  circle(0.5);
				\draw [xshift=0.1cm] (29,3.6)  node[anchor=south] {$\zeta$}  (28.5,1.5)  node {\bf{ }};
				\draw [xshift=0.1cm] (28,0.7)   node {\bf{ }};
				\draw [xshift=0.1cm] (30.3,0.7)   node {$y$};
				\draw[  xshift=0.1cm,  >=stealth] [->] (29.5,0) -- (31.5,0);
				\draw[  xshift=0.1cm,  >=stealth] [->] (29,3.7) -- (29,0.5);
				\draw[ xshift=0.1cm,  >=stealth] [->] (29,-0.5) -- (29,-3) -- (4,-3)-- (4, -0.5);
				\draw [xshift=0.1cm] (3.3,-1.3)   node {\bf{--}};
				\useasboundingbox (0,0.1);
			\end{tikzpicture}
			\caption{Feedback loop of the plant $\bf G$ with the controller $\bf K$\vspace{-3mm}}
			\label{2Block}
		\end{figure}
		
		\subsection{Standard Unity Feedback} 


	We focus on the standard unity feedback   of Fig.~\ref{2Block}, where ${\bf G}\in \mathbb{R}_p(\la)^{p \times m}$ is the LTI plant and ${\bf K}\in \mathbb{R}_p(\la)^{m \times p}$ is its LTI controller. Here $r$, $w$ and $\zeta$ are the reference signal, input disturbance and sensor noise vectors, respectively, while $y$, $u$, $z$ and $v$ are the measurement, command, regulated and applied control signal vectors, respectively.  
	If all the closed-loop maps from the exogenous signals $\begin{bmatrix}
		r^\top & w ^\top & \zeta^\top \end{bmatrix}^\top$ to $\begin{bmatrix}
		y^\top & u^\top & z^\top & v^\top \end{bmatrix}^\top$, \emph{i.e.}, any point inside the feedback loop of Fig.~\ref{2Block}, are stable then we say that ${\bf K}$ is an (internally)  stabilizing controller of ${\bf G}$ or that ${\bf K}$ (internally) stabilizes ${\bf G}$.


	\subsection{The Youla Parametrization}\label{2}
	
	The seminal work presented in \cite{V} highlighted the fact that the problem of closed-loop stabilization can be conveniently reformulated in terms of particular fractional representations, belonging to the TFMs which make up the feedback control scheme. In this context, we introduce several key definitions.
	
	\begin{defn} \label{14Martie2019} 
		Let ${\bf K} \in \mathbb{R}_p(\la)^{m \times p}$, ${\bf R}\in \mathbb{R}_p(\la)^{m \times m}$, ${\bf P} \in \mathbb{R}_p(\la)^{m \times p}$. A representation of the form ${\bf K}={\bf R}^{-1}{\bf P}$ is called a {\em left factorization} of {\bf K}.
	\end{defn}
	
	\begin{rem}\label{rem:inv}
		If ${\bf K}={\bf Y}^{-1}{\bf X}$ is a left factorization of ${\bf K}$, then  any other left factorization of ${\bf K}$ (as in Definition~\ref{14Martie2019}) 
		is of the form ${\bf R}={\bf UY}$, ${\bf P}={\bf UX}$, for some invertible TFM ${\bf U}$.
	\end{rem}
	
	\begin{defn}[Corollary 4.1.4 in \cite{V}]
		A left factorization ${\bf G} = \tilde{\bf M}^{-1}\tilde{\bf N}\in \mathbb{R}_p(\la)^{p \times m}$ with $\tilde{\bf N}\in \mathbb{R}_{\mathbb{S}}(\la)^{p \times m}$, $ \tilde {\bf M} \in \mathbb{R}_{\mathbb{S}}(\la)^{p \times p} $ is additionally called \emph{coprime over $\mathbb{R}_{\mathbb{S}}(\la)$} if there exist $\tilde{ \bf X} \in \mathbb{R}_{\mathbb{S}}(\la)^{m \times p}$,\hspace{-1mm} $\tilde {\bf Y} \in \mathbb{R}_{\mathbb{S}}(\la)^{p \times p}$ so that $\tilde{\bf M} \tilde{\bf Y} + \tilde{\bf N}\tilde{\bf X} =I_p$.
	\end{defn}

	\begin{rem}
		Note also that a \emph{right coprime factorization over} $\mathbb{R}_{\mathbb{S}}(\la)$ of $\mathbf{G}$ can be obtained in a straightforward fashion by simply transposing the left coprime factors over $\mathbb{R}_{\mathbb{S}}(\la)$ of $\mathbf{G}^\top$.
	\end{rem}
	
	
	\begin{defn}(Remark 4.1.17 in \cite{V})  
		A collection of eight  stable TFMs $\big({\bf M}, {\bf  N}, \tilde {\bf  M}, \tilde {\bf  N}, {\bf X}, {\bf  Y}, \tilde {\bf  X}, \tilde {\bf  Y}\big)$ is called a  {\em doubly coprime factorization} (DCF) over $\mathbb{R}_{\mathbb{S}}(\lambda)$  of ${\bf  G}\in \mathbb{R}_p(\la)^{p \times m}$  if   $\tilde {\bf  M}$ and ${\bf M}$ are both invertible, they yield the
		factorizations
		${\bf G}=\tilde{\bf M}^{-1}\tilde{\bf N}={\bf NM}^{-1}$ and they satisfy the following equality
		\begin{equation}\label{dcrel}
			\small\ba{cc}   \phantom{-}{\bf Y} & {\bf X}  \\ - \tilde {\bf N} & \tilde {\bf M}  \ea
			\small\ba{cc}  {\bf M} & -\tilde {\bf X} \\   {\bf N} &  \phantom{-}\tilde {\bf Y} \ea = I_{m+p}.\normalsize
		\end{equation}
	\end{defn}
	
	The next theorem provides, via such DCFs over $\mathbb{R}_{\mathbb{S}}(\lambda)$, the  parameterization of \emph{all} stabilizing controllers for a given plant.

	\begin{theorem}[Theorem 5.2.1 in \cite{V}]
		\label{Youlaaa}
		 Let  $\big({\bf M}, {\bf  N}$, $\tilde {\bf  M}, \tilde {\bf  N}$, ${\bf X}, {\bf  Y}$, $\tilde {\bf  X}, \tilde {\bf  Y}\big)$ be a DCF over $\mathbb{R}_{\mathbb{S}}(\lambda)$ of ${\bf G}\in \mathbb{R}_p(\la)^{p \times m}$. Define
		\begin{equation} \label{EqYoula4}
			\begin{split}
				{\bf X_Q}  :=  {\bf X}+{\bf Q} \tilde{\bf M}, \quad
				&{\bf \tilde{X}_Q}  :=  \tilde{\bf X}+{\bf MQ}, \quad \\
				{\bf Y_Q}  :=  {\bf Y} - {\bf Q} \tilde{\bf N}, \quad 
				&{\bf \tilde{Y}_Q}  :=  \tilde{\bf Y}-{\bf NQ},  
			\end{split}
		\end{equation}
		for some ${\bf Q}\in\mathbb{R}_{\mathbb{S}}(\la)^{m\times p}$, and note also that the following generalization of the identity from \eqref{dcrel} holds
		\begin{equation}\label{dcrelQ}
			\! \small\ba{cc} \!  {\bf Y_Q} &   {\bf X_Q} \\  -\tilde {\bf N}  &  \tilde {\bf M}  \! \ea \!
			\small\ba{cr} \!  {\bf M}   & -\tilde {\bf X}_{\bf Q}   \\   {\bf N}   & \tilde {\bf Y}_{\bf Q}  \!  \ea = I_{m+p}\ .
			\normalsize
		\end{equation}
		Then, 
		the class of all controllers ${\bf K_Q}$ which stabilize the plant ${\bf G}$ (in feedback interconnection, see Fig.~\ref{2Block}) is given by
			\begin{equation}
				\label{YoulaEq}
				{\bf K_Q}={\bf Y}^{-1}_{\bf Q}{\bf X_Q} = {\bf \tilde{X}_Q}  {\bf \tilde{Y}}^{-1}_{\bf Q}\in\mathbb{R}(\la)^{m\times p},
			\end{equation}
			for all ${\bf Q}\in\mathbb{R}_{\mathbb{S}}(\la)^{m\times p}$ which ensure that both ${\bf Y_Q}$ and ${\bf \tilde{Y}_Q}$ are invertible TFMs.
			
	\end{theorem}

	
	\begin{figure*}
		\begin{equation} \label{klap1}
			\mathbf{H}_{CL}({\bf G}, {\bf K}_{{\bf Q}}):=\footnotesize\ba{ccc}  (I_p+{\bf GK_Q})^{-1}{\bf GK_Q} & (I_p+{\bf GK_Q})^{-1}{\bf G} &(I_p+{\bf GK_Q})^{-1}  \\ 
			(I_m+{\bf K_Q}{\bf G})^{-1}{\bf K_Q} & -(I_m+{\bf K_QG})^{-1}{\bf K_QG} &-(I_m+{\bf K_Q}{\bf G})^{-1}{\bf K_Q} \\
			(I_p+{\bf GK_Q})^{-1} & -(I_p+{\bf GK_Q})^{-1}{\bf G} & -(I_p+{\bf GK_Q})^{-1} \\
			(I_m+{\bf K_QG})^{-1} {\bf K_Q}& (I_m+{\bf K_QG})^{-1} & - (I_m+{\bf K_Q G})^{-1}{\bf K_Q}
			\ea\normalsize\tag{6}
		\end{equation}\hrulefill\end{figure*}
	
	Denote by $\mathbf{H}_{CL}({\bf G}, {\bf K}_{{\bf Q}})$ the TFM from $\begin{bmatrix}
		r^\top & w ^\top & \zeta^\top \end{bmatrix}^\top$ to $\begin{bmatrix}
		y^\top & u^\top & z^\top & v^\top \end{bmatrix}^\top$, whose entries are the achievable closed-loop maps produced by stabilizing controllers (\ref{YoulaEq}) and which are given explicitly in (\ref{klap1}), at the top of the next page. One of the chief features of the Youla Parameterization is the fact that it renders all closed-loop maps from the feedback loop in Fig.~\ref{2Block} as \emph{affine expressions} of the free (and stable) parameter $\mathbf{Q}$, as highlighted via the following result.
	
	\begin{coro}[Corollary 5.2.3 in \cite{V}] \label{afinecloop}
		 The set of {\em all} closed-loop maps (\ref{klap1}) achievable via stabilizing controllers (\ref{YoulaEq}) are affine in the Youla parameter ${\bf Q}$ and are, moreover, given by \stepcounter{equation}
		\begin{equation} \label{finallly}
			\small\begin{tabular}{|c |  c|  c|  c|} %
				\hline
				${} $ & ${r}$ & ${w}$ & ${\zeta}$ \\
				\hline
				$y$  & ${\bf N} {\bf X}_{\bf Q}$ & ${\bf N} {\bf Y}_{\bf Q}$  & $I_p- {\bf N} {\bf X}_{\bf Q}$   \\ 
				\hline
				$u$ & ${\bf M} {\bf X}_{\bf Q}$ & ${\bf M} {\bf Y}_{\bf Q}-I_m$ & $-{\bf M} {\bf X}_{\bf Q}$  \\[0.5ex]
				\hline
				$z $  & $I_p-{\bf N} {\bf X}_{\bf Q}$ & $-{\bf N} {\bf Y}_{\bf Q}$ &${\bf N} {\bf X}_{\bf Q}-I_p$ \\[0.5ex]
				\hline
				$ v$  & ${\bf M} {\bf X}_{\bf Q}$ & ${\bf M} {\bf Y}_{\bf Q}$ & $ -{\bf M} {\bf X}_{\bf Q}$ \\
				\hline
			\end{tabular}\normalsize
		\end{equation}
	\end{coro}

	\subsection{Network Realization Functions}
	
	\begin{figure}
		\includegraphics[scale=.7]{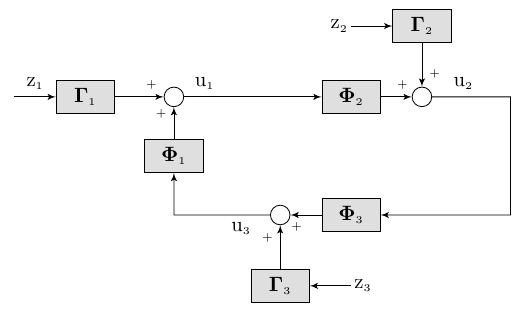}
		\caption{Three-hop network architecture\vspace{-5mm}}
		\label{fig:three_hop}
	\end{figure}

	For descriptive simplicity, we focus on the three-hop ``ring'' network from Fig.~\ref{fig:three_hop} and we describe its signal flow equations
	\begin{equation*} \label{firstimpression}
		\small\ba{c} \hspace{-1mm}u_1 \\ \hspace{-1mm}u_2 \\\hspace{-1mm}u_3 \ea  \hspace{-1.5mm}=\hspace{-1.5mm} \small\ba{ccc} O & \hspace{-2mm}O & \hspace{-2mm}{\bf \Phi}_1 \\ 
		{\bf \Phi}_2  & \hspace{-2mm}O  & \hspace{-2mm}O\\ O &  \hspace{-2mm}{\bf \Phi}_3  &\hspace{-2mm}O\ea\hspace{-2mm} \small\ba{c} \hspace{-1mm}u_1 \\ \hspace{-1mm}u_2 \\ \hspace{-1mm}u_3 \ea \hspace{-1mm}+ 
		\hspace{-1mm}\small\ba{ccc}{\bf \Gamma}_{1} & \hspace{-2mm}O&\hspace{-2mm}O\\ 
		O & \hspace{-2mm}{\bf \Gamma}_{2} &\hspace{-2mm}O\\ O&\hspace{-2mm}O& \hspace{-2mm}{\bf \Gamma}_{3} \ea\hspace{-2mm}  \small\ba{c} \hspace{-1mm}z_1 \\ \hspace{-1mm}z_2\\\hspace{-1mm}z_3 \ea\hspace{-1mm}. \normalsize
	\end{equation*} 
	Notice that the three by three  ${\bf \Phi}$ factor takes the precise meaning of the network's directed graph {\em adjacency matrix}, with the LTI filters ${\bf \Phi}_1, {\bf \Phi}_2$ and ${\bf \Phi}_3$, respectively, having the significance of  weights of their corresponding edges. The remaining three by three ${\bf \Gamma}$ factor has the  role of defining the {\em input terminals} of the network, {\em i.e.}, the points of access (to the network) of the exogenous signals $z_1, z_2$ and $z_3$, respectively.  	
	
	Assuming that $(I_m-{\bf \Phi})$ is invertible, {\em i.e.}, the LTI network from Fig.~\ref{fig:three_hop} is well-posed, the ensuing left factorization   
	\begin{equation} \label{secondimpression}
		\small\ba{c} u_1 \\ u_2 \\u_3 \ea \hspace{-1.5mm} = \hspace{-1.5mm} \small\ba{ccc} \hspace{-2mm}I & \hspace{-2mm}O & \hspace{-2mm}-{\bf \Phi}_1\hspace{-2mm} \\ 
		\hspace{-2mm}-{\bf \Phi}_2  & \hspace{-2mm}I  & \hspace{-2mm}O\hspace{-3mm}\\\hspace{-2mm} O &  \hspace{-2mm}-{\bf \Phi}_3 &\hspace{-2mm}I\hspace{-3mm}\ea ^{\hspace{-0.5mm}-1}\hspace{-1.5mm} \small\ba{ccc}{\hspace{-1mm}\bf \Gamma}_{1} & \hspace{-2mm}O&\hspace{-2mm}O\\ 
		\hspace{-1mm}O & \hspace{-2mm}{\bf \Gamma}_{2} &\hspace{-2mm}O\\ \hspace{-1mm}O&\hspace{-2mm}O& \hspace{-2mm}{\bf \Gamma}_{3} \ea\hspace{-2mm}  \small\ba{c} z_1 \\ z_2\\z_3 \ea\normalsize \normalsize\hspace{-1mm}\normalsize
	\end{equation}
	yields  the Input/Output map (from  $ z$ to ${u}$), which we denote by ${\bf K}$. Note that, in general, the sparsity patterns of the ${\bf \Phi}$ and ${\bf \Gamma}$ factors are completely lost in (\ref{secondimpression}) due to the inversion of  $(I_m - {\bf \Phi})$ that generically yields a ``full'' TFM ${\bf K}$, with no particular sparsity pattern of its own. The distinctive ``structure'' of the LTI network in Fig.~\ref{fig:three_hop}, as captured by the $({\bf \Phi}, {\bf \Gamma})$ pair, cannot in general be retrieved solely from ${\bf K}$ (see also \cite{VanDenHof}).
	
	\begin{rem}
		This type of architecture has been successfully employed for the control laws proposed in \cite{TAC2016}, with the aim of eliminating downstream disturbance propagation for a platoon of autonomous vehicles. Notably, guarantees of closed-loop norm optimality were also obtained in Section V-B of \cite{TAC2016}.
	\end{rem}

	\begin{defn} \label{NRFDEF} Given ${\bf K}\in\mathbb{R}_p(\lambda)^{m \times p}$, ${\bf \Phi}\in\mathbb{R}_p(\lambda)^{m \times m}$ and ${\bf \Gamma}\in\mathbb{R}_p(\lambda)^{m \times p}$ such that ${\bf \Phi}$ has all its diagonal entries equal to zero and ${\bf K} = (I_m-{\bf \Phi})^{-1} {\bf \Gamma}$, the pair of TFMs $({\bf \Phi}, {\bf \Gamma})$ is said to be a {\em Network Realization Function} pair of $\mathbf{K}$.
	\end{defn}
	
	\begin{rem}\label{hollow}
		Notice the fact that any NRF is ultimately a left factorization of ${\bf K}$. For any ${\bf K}\in\mathbb{R}_p(\lambda)^{m \times p}$, let ${\bf K}={\bf R}^{-1}{\bf P}$ be some left factorization (recall Definition~\ref{14Martie2019}). In this case, the gain at infinity of the ``denominator'' TFM, ${\bf R}$, and that of its diagonal component, ${\bf R}^{\text{\em diag}}$, can always be made (recall Remark \ref{rem:inv}) equal to the identity matrix. Therefore, ${\bf R}^{\text{\em diag}}$ will have a proper inverse, from which we get that
		\begin{equation}\label{NRFmain}
			\displaystyle \Big( {\bf \Phi} :=I_m -({\bf R}^{\text{\em diag}})^{-1}{\bf R}\: , \: \; \; \; \; {\bf \Gamma} :=({\bf R}^{\text{\em diag}})^{-1}{\bf P}\Big)
		\end{equation}
		satisfies Definition~\ref{NRFDEF}, making it an NRF pair of ${\bf K}$. Notably, the transformation from \eqref{NRFmain} preserves the sparsity patterns: ${\bf \Phi}$ retains the sparsity pattern of ${\bf R}$ while ${\bf \Gamma}$ retains that of ${\bf P}$.
	\end{rem}
	

	\section{Distributed Control via NRF Implementation}\label{sec:NRF}
	
	\subsection{Specifying Sensing and Communication Constraints}\label{subsec:spec}

	The stated aim of this paper is to investigate distributed implementations of output feedback  controllers as  networks of LTI filters. 
	In the NRF framework, the control law
	\begin{equation}
		\label{imple}
		\begin{array}{cccc}
			u= & \underbrace{ \quad {\bf \Phi}\: u \quad} \:\: & + &\underbrace{\quad {\bf \Gamma}\: z \quad}\ , \\
			& \textrm{feedforward} & \: &\textrm{feedback\phantom{\ ,}}
		\end{array}
	\end{equation}
	which bears striking resemblance to the architecture proposed in (16) from \cite{distrib_ACC},	has a twofold manifestation: firstly, in {\em the sparsity pattern of  the} ${\bf \Phi}$ {\em factor}, by designating which control signals are available, and secondly, in {\em that of the} ${\bf \Gamma}$ {\em factor}, by defining which of the regulated measurements are available.

	
	The communication constraints ${\bf \Phi}\in \mathcal{Y}$ are imposed on the distributed controller by way of pre-specifying the linear subspace $\mathcal{Y} \subseteq \mathbb{R}_p (\la)^{m \times m}$, while the sensing constraints ${\bf \Gamma}\in \mathcal{X}$  are  encapsulated  in the   pre-specified  linear subspace $\mathcal{X}\subseteq \mathbb{R}_p (\la)^{m \times p}$, respectively. 
	The  subspace $\mathcal{Y}^+$ is obtained by  allowing 
	for non-zero diagonal entries on the elements from $\mathcal{Y}$ such that ${\bf Y}\in \mathcal{Y}^+$  $\iff  ({\bf Y}-{\bf Y}^{\text {diag}}) \in \mathcal{Y}$.

	\begin{rem} \label{scale} 
		In the NRF framework, we avoid communicating \emph{internal states}, {\em i.e.}, states of the plant or controller, thus promoting control law implementations that are scalable with respect to the dimension of the plant's state (recall Section \ref{Introdc}).
	\end{rem}
	

	\subsection{Internal Stability Guarantees}\label{subsec:IOstab}

	In the NRF-based  implementation (\ref{imple}) of the controller,  the variable $u$ may be affected by the additive disturbance denoted $\delta_u$, with the equation for the controller from Fig.~\ref{2BlockAgain}  reading as
	\begin{equation}\label{impletrue}
		u={\bf \Phi} (u+\delta_u)+ {\bf \Gamma} z. 
	\end{equation}
	The internal stability analysis must certify that the closed-loop maps from $\delta_u$ to the signals $z,u,v$ and $y$ are all stable.\vspace{-1mm}

	\begin{assumption}\label{as:sp}
		The plant ${\bf G}$ is strictly proper. \vspace{-1mm}
	\end{assumption}\newpage
	
	\begin{rem}\label{rem:scale}
		We point out that Assumption \ref{as:sp} is by no means restrictive, as shown via \cite{aug_sparse}, and has been made only to facilitate the presentation of the NRF design formalism. For example, a direct consequence of Assumption \ref{as:sp} is the fact that the TFMs  ${\bf M, Y, \tilde M, \tilde Y}$ from {\em any} DCF of ${\bf G}$ can be scaled in (\ref{dcrel}) to make their gain at infinity  equal to the identity matrix. Thus, all DCFs of type \eqref{EqYoula4}-\eqref{dcrelQ} which will be employed in the sequel are taken to have the aforementioned property, implying that $\big( \mathbf{Y}_{\mathbf{Q}}^{{\text {\em diag}}}\big)^{-1}\in\mathbb{R}_p(\la)^{m \times m}$ for any $\mathbf{Q}\in\mathbb{R}_{\mathbb{S}}(\la)^{m \times p}$.
	\end{rem} 
	
	\begin{figure}[t]
		\begin{tikzpicture}[scale=0.27]
			\draw[xshift=0.1cm, >=stealth ] [->] (0,0) -- (3.5,0);
			\draw[ xshift=0.1cm ]  (4,0) circle(0.5);
			\draw[xshift=0.1cm] (3,1)   node {\bf{ }} (1,0.8) node {$r$};
			\draw [xshift=0.1cm](6,0.8)   node {$z$} ;
			\draw[ xshift=0.1cm,  >=stealth] [->] (4.5,0) -- (8,0);
			\draw[ thick, xshift=0.1cm]  (8,-1.5) rectangle +(3,3);
			\draw [xshift=0.1cm](9.5,0)   node {{$\hspace{0.5mm}\mathbf{K}_{\mathbf{Q}}$}} ;
			\draw[ xshift=0.1cm,  >=stealth] [->] (11,0) -- (15.5,0);
			\draw[ xshift=0.1cm ]  (16,0) circle(0.5cm);
			\draw [xshift=0.1cm](12.6,0.8)   node {$u$} ;
			\draw[  xshift=0.1cm,  >=stealth] [->] (14,0) -- (14,4)--(10,4);
			\draw [xshift=0.1cm](18.6,0.8)   node {$v$} ;
			\draw [xshift=0.1cm](14,0)   node {$\bullet$} ;
			\draw [xshift=0.1cm] (14.8,0.7)   node {\bf{ }};
			\draw[  xshift=0.1cm,  >=stealth] [->] (16,3.7) -- (16,0.5);
			\draw [xshift=0.1cm] (0.8,4)  node[anchor=south] {$\delta_u$};
			\draw[ xshift=0.1cm ]  (9.5,4) circle(0.5);
			\draw[  xshift=0.1cm,  >=stealth] [->] (9.5,3.5) -- (9.5,1.5);
			\draw[  xshift=0.1cm,  >=stealth] [->] (0,4) -- (9,4);
			\draw [xshift=0.1cm] (16,3.6)  node[anchor=south] {$w$}  (15.3,1.5)  node {\bf{ }};
			\draw[  xshift=0.1cm,  >=stealth] [->] (16.5,0) -- (21,0);
			\draw[ thick, xshift=0.1cm ]  (21,-1.5) rectangle +(3,3) ;
			\draw [xshift=0.1cm] (22.5,0)   node {{${\bf G}$}} ;
			\draw[ xshift=0.1cm,  >=stealth] [->] (24,0) -- (28.5,0);
			\draw[ xshift=0.1cm ] (29,0)  circle(0.5);
			\draw [xshift=0.1cm] (29,3.6)  node[anchor=south] {$\zeta$}  (28.5,1.5)  node {\bf{ }};
			\draw [xshift=0.1cm] (28,0.7)   node {\bf{ }};
			\draw [xshift=0.1cm] (30.3,0.7)   node {$y$};
			\draw[  xshift=0.1cm,  >=stealth] [->] (29.5,0) -- (31.5,0);
			\draw[  xshift=0.1cm,  >=stealth] [->] (29,3.7) -- (29,0.5);
			\draw[ xshift=0.1cm,  >=stealth] [->] (29,-0.5) -- (29,-3) -- (4,-3)-- (4, -0.5);
			\draw [xshift=0.1cm] (3.3,-1.3)   node {\bf{--}};
			\useasboundingbox (0,0.1);
		\end{tikzpicture}
		\caption{Feedback loop of the plant ${\bf G}$ with the controller  ${\bf K}$ in an NRF-based implementation $u={\bf \Phi} (u+\delta_u)+ {\bf \Gamma} z$\vspace{-6mm}}
		\label{2BlockAgain}
	\end{figure}
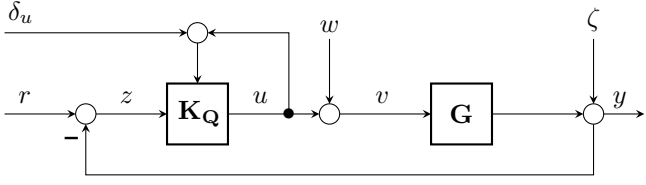
	
	The next theorem shows that Remark~\ref{hollow} offers a natural mechanism to obtain stabilizing NRF-based implementations.

	\begin{theorem} \label{IS}
		Let ${\bf G}\in\mathbb{R}_p(\lambda)^{p\times m}$ be given by one of its DCFs, as in \eqref{dcrel}, and define the $\mathbf{Q}$-parameterized NRF pairs\vspace{-1mm}
		\begin{subequations}
			\begin{equation}
				\label{ISmain1}
				{\bf \Phi}:={I_m-\big(\bf Y_Q^{{\text {\em diag}}}\big)}^{-1}{\bf Y_Q},\vspace{-2mm}
			\end{equation} 
			\begin{equation}
				\label{ISmain2}
				\hspace{-8.5mm}{\bf \Gamma}:={\big( \bf Y_Q^{{\text {\em diag}}}\big)}^{-1}{\bf X_Q},\vspace{-1mm}
			\end{equation} 
		\end{subequations}
		associated with a $\mathbf{K}_{\mathbf{Q}}$, as in \eqref{YoulaEq}. Then, we have that:
		\begin{enumerate}
			\item[$\mathbf{(a)}$] The NRF-based implementation of ${\bf K_Q}$ from \eqref{impletrue} internally stabilizes the feedback loop shown in Fig. \ref{2BlockAgain};
			
			\item[$\mathbf{(b)}$] ${\bf \Phi}\in \mathcal{Y}$, ${\bf \Gamma}\in \mathcal{X}$ if and only if ${\bf Y_Q}\in \mathcal{Y^+}$, ${\bf X_Q}\in \mathcal{X}$.
		\end{enumerate}
	\end{theorem}
	\begin{proof}
		For point $\mathbf{(a)}$, see the Appendix. Point $\mathbf{(b)}$ follows directly from the arguments made in Remark \ref{hollow}.
	\end{proof}

	\vspace{-6mm}
	\subsection{Realization-based Distributed Implementations}
	
	Distinguishing to the NRF setup, the main result of this section shows that the closed-loop state dynamics from Fig.~\ref{2BlockAgain} can be made stable in a distributed fashion, by implementing stabilizable and detectable realizations for each row of $\begin{bmatrix}
		\mathbf{\Phi}&\mathbf{\Gamma} 
	\end{bmatrix}$. Before extending the input-output stability guarantees given in Section~\ref{subsec:IOstab}, we first require the next preparatory lemma.
	
	
	
	\begin{lem}\label{lem:pole_conserv}
		Let $\mathbf{G}_1$ and $\mathbf{G}_2$ be two proper TFMs, with $\mathbf{G}_2$ stable and having full row normal rank along with no transmission zeros in $\mathbb{C}\backslash\mathbb{S}$. Then, $\mathcal{P}_u(\mathbf{G}_1\mathbf{G}_2)=\mathcal{P}_u(\mathbf{G}_1)$.\vspace{-1mm}
	\end{lem}
	
	\begin{proof}
		See the Appendix.
	\end{proof}
	
	We now formulate the result that shows how to implement distributed and realization-based versions of \eqref{impletrue}, which successfully stabilize the closed-loop state dynamics from Fig.~\ref{2BlockAgain}.
	
	\begin{theorem}\label{thm:implem}
		Let $\mathbf{G}\in\mathbb{R}_p(\lambda)^{p\times m}$ be given by a stabilizable and detectable realization of type \eqref{ss0ab}-\eqref{ss0c} and let $\mathbf{K}\in\mathbb{R}_p(\lambda)^{m\times p}$ be an internally stabilizing controller of $\mathbf{G}$, which is described by an NRF pair $(\mathbf{\Phi},\mathbf{\Gamma})$ as in \eqref{ISmain1}-\eqref{ISmain2}. Then, by implementing stabilizable and detectable realizations for each $e_i^\top\begin{bmatrix}
			\mathbf{\Phi}&\mathbf{\Gamma}
		\end{bmatrix}$, with $i\in1:m$, and by computing the commands as in \eqref{impletrue}, the closed-loop state dynamics of the system from Fig.~\ref{2BlockAgain} will be asymptotically stable (see Section 5.3 in \cite{zhou}).\vspace{-1mm}
	\end{theorem}
	
	\begin{proof}
		See the Appendix.
	\end{proof}

	\begin{rem}\label{rem:implem}
		By obtaining minimal state-space realizations denoted $e_i^\top\begin{bmatrix}
			\mathbf{\Phi} & \mathbf{\Gamma}
		\end{bmatrix}=\left[\footnotesize\begin{array}{c|c}
			A_i-\lambda I_{n_i}& B_i\\\hline C_i&D_i
		\end{array}\right],\ \forall i\in1:m$, notice that each zero column of $e_i^\top\begin{bmatrix}
			\mathbf{\Phi} & \mathbf{\Gamma}
		\end{bmatrix}$ will result in a zero column on the same position in both $B_i$ and $D_i$, producing structured and stabilizing state-space-based implementations.
		
	\end{rem}
	
	\begin{rem}\label{rem:scalab}
		Quite notably, Remark \ref{rem:implem} and Theorem \ref{thm:implem} still hold when implementing minimal and, respectively, stabilizable and detectable realizations for block-rows of $\begin{bmatrix}
			\mathbf{\Phi}&\mathbf{\Gamma}
		\end{bmatrix}$, instead of just single rows. Moreover, minimal realizations of such block-rows may also reduce the number of sub-controller states at any location where more than a single command signal is computed, emphasizing the scalability of our method.
	\end{rem}
	
	\noindent We conclude this subsection with the next consequence of Theorem \ref{thm:implem}, which certifies closed-loop stability with respect to bounded additive disturbance in dynamics of type \eqref{ss0ab}-\eqref{ss0c}.
	
	\begin{coro}\label{cor:delta}
		Let the same hypotheses and notation hold as in the statement of Theorem \ref{thm:implem}. Let the state dynamics \eqref{ss0ab} of $\mathbf{G}$ and $\begin{bmatrix}
			\mathbf{\Phi}&\mathbf{\Gamma}
		\end{bmatrix}$ be affected additively by the bounded disturbances $\delta_{x_{\mathbf{G}}}$ and $\delta_{x_{\mathbf{K}}}$, respectively, and the output dynamics \eqref{ss0c} of $\mathbf{G}$ and $\begin{bmatrix}
		\mathbf{\Phi}&\mathbf{\Gamma}
	\end{bmatrix}$ be affected additively by the bounded disturbances $\delta_{y_{\mathbf{G}}}$ and $\delta_{y_{\mathbf{K}}}$, respectively. Then, all the closed-loop signals and state variables in Fig. \ref{2BlockAgain} remain bounded.
	\end{coro}
	\begin{proof}
		See the Appendix.
	\end{proof}
	
	\begin{rem}
		Similarly to the SLS framework from \cite{matni1,matni2}, we consider the bounded disturbances $\delta_{x_{\mathbf{K}}}$ and $\delta_{y_{\mathbf{K}}}$ as arising from computational errors and the imperfect implementation of the controller's NRF pair. Moreover, the bounded disturbances $\delta_{x_{\mathbf{G}}}$ and $\delta_{y_{\mathbf{G}}}$ can be attributed to unmapped network dynamics and sources of disturbance which are not captured in \eqref{ss0ab}-\eqref{ss0c}.
	\end{rem}
	
	\subsection{Norm-based Optimal Design}
	
	Given the sensing and communication subspace constraints  $\mathcal{X}$ and $\mathcal{Y}$, we desire to obtain NRF-based implementations of distributed controllers ${\bf K}$ which solve the following problem
	\begin{subequations}
		\begin{alignat}{3}
			&\!\min_{\Huge{{\bf K}}\in\mathbb{R}(\lambda)^{m\times p}}        &\qquad& \norm{\: \mathbf{H}_{CL}({\bf G},{\bf K}) \:}\label{eq:optProbA}\\
			&\text{subject to} &      & \text{Fig. \ref{2BlockAgain} is internally stable} ,\label{eq:constraintA1}\\
			&                  &      & {\bf K} = (I_m-{\bf \Phi})^{-1} {\bf \Gamma}, \label{eq:constraintA2} \\
			&                  &      & {\bf \Phi}\in \mathcal{Y}, \; {\bf \Gamma}\in \mathcal{X}. \label{eq:constraintA3}
		\end{alignat}
	\end{subequations}
	
	

	With the difficulty of (\ref{eq:optProbA})--(\ref{eq:constraintA3}) being well-understood in literature, and its epitome being the computation of an optimal controller having a (block-)diagonal TFM, we now focus on a tractable adaptation of it. The latter is given via the following result, 
	which has the benefit of being stated in terms of affine expressions, starting from a {\em fixed} DCF \eqref{dcrel} of the plant.

	\begin{coro} \label{July20primo}
		Let ${\bf G}\in {\mathbb{R}_p(\la)}^{p\times m}$. Consider $(\mathbf{\Phi},\mathbf{\Gamma})$ from \eqref{ISmain1}-\eqref{ISmain2}, based upon a DCF \eqref{dcrel} of $\mathbf{G}$. Then, (\ref{eq:optProbA})--(\ref{eq:constraintA3}) is equivalent to the following affine model matching problem
		\begin{subequations} \label{MReq}
			\begin{alignat}{4}
				&\!\min_{\Huge{{\bf Q}\in\mathbb{R}_{\mathbb{S}}(\lambda)^{m\times p}}}        &\qquad& \norm{ \: \mathbf{H}_{CL}({\bf G},{\bf K_Q}) \:}\label{eq:optProbC}\\
				&\emph{subject to} &      & {\bf Y_Q}\in \mathcal{Y^+}, \; {\bf X_Q}\in \mathcal{X}.   \label{eq:constraintC2} 
			\end{alignat}
		\end{subequations}
	\end{coro}
	\begin{proof}
		The result follows directly from Theorem \ref{IS}.
	\end{proof}
	
	\begin{rem}
		Efficient numerical solutions for type \eqref{eq:optProbC}-\eqref{eq:constraintC2} problems were proposed in \cite{aug_sparse}. However, an important limitation of Corollary~\ref{July20primo} is that its outcome depends on the initial choice of a DCF over $\mathbb{R}_{\mathbb{S}}(\lambda)$ for the plant. This was to be expected and has been alleviated in part by the subsequent results from \cite{aug_sparse}. Similarly, the outcome of the SLS \cite{matni2} depends on the initial choice of a realization of the plant. Moreover, we point out that the suboptimality gap induced by the desired sparsity structure of $\mathbf{\Phi}$ and $\mathbf{\Gamma}$ must also be taken into account in the design phase, as is the case with the SLS (see, for example, \cite{reg} and the numerical examples from \cite{matni1}).
	\end{rem}
	

	\vspace{-4mm}
	
	\section{Alternative Representations}\label{sec:alt}

	\begin{figure}
		\centering
		\begin{tikzpicture}[scale=.9, every node/.style={transform shape}]
			\node(n1)[minimum height=0.5cm, minimum width=0.5cm,draw]at(0,0) {Node $1$};
			\node(n2)[minimum height=0.5cm, minimum width=0.5cm,draw]at(-3.5,-0.8) {Node $2$};
			\node(n3)[minimum height=0.5cm, minimum width=0.5cm,draw]at(-3.5,0.8) {Node $3$};
			\node(n4)[minimum height=0.5cm, minimum width=0.5cm,draw]at(3.5,-0.8) {Node $4$};
			\node(n5)[minimum height=0.5cm, minimum width=0.5cm,draw]at(3.5,0.8) {Node $5$};
			
			\node(n11)[above of=n1, node distance=3em]{$u_1$};
			\draw[-latex](n11)--(n1);
			\node(n12)[below of=n1, node distance=3em]{$y_1$};
			\draw[-latex](n1)--(n12);
			
			\node(n31) at (-3.8, 1.75){$u_3$};
			\draw[-latex](n31)--(-3.8, 1.05);
			\node(n32) at (-3.2, 1.75){$y_3$};
			\draw[-latex](-3.2, 1.05)--(n32);
			
			\node(n51) at (3.2, 1.75){$u_5$};
			\draw[-latex](n51)--(3.2, 1.05);
			\node(n52) at (3.8, 1.75){$y_5$};
			\draw[-latex](3.8, 1.05)--(n52);
			
			\node(n21) at (-3.8, -1.75){$u_2$};
			\draw[-latex](n21)--(-3.8, -1.05);
			\node(n22) at (-3.2, -1.75){$y_2$};
			\draw[-latex](-3.2, -1.05)--(n22);
			
			\node(n41) at (3.2, -1.75){$u_4$};
			\draw[-latex](n41)--(3.2, -1.05);
			\node(n42) at (3.8, -1.75){$y_4$};
			\draw[-latex](3.8, -1.05)--(n42);
			
			\draw[-latex](n1)--(n2) node[near end, below]{$y_{1}$};
			\draw[-latex](n1)--(n3) node[near end, above]{$y_{1}$};
			\draw[-latex](n1)--(n4) node[near end, below]{$y_{1}$};
			\draw[-latex](n1)--(n5) node[near end, above]{$y_{1}$};
			\draw[-latex](n2)--(n3) node[midway, left]{$y_{2}$};
			
			\draw[dashed](-1.5,2.75)--(-1.5,-2.25);
			\draw[dashed](1.5,2.75)--(1.5,-2.25);
			\node at (-3.5,2.5){Area 2};
			\node at (0,2.5){Area 1};
			\node at (3.5,2.5){Area 3};
		\end{tikzpicture}
		\caption{Interconnection of the network's various nodes and the areas of admissible communication\vspace{-4mm}}\label{fig:area}
	\end{figure}
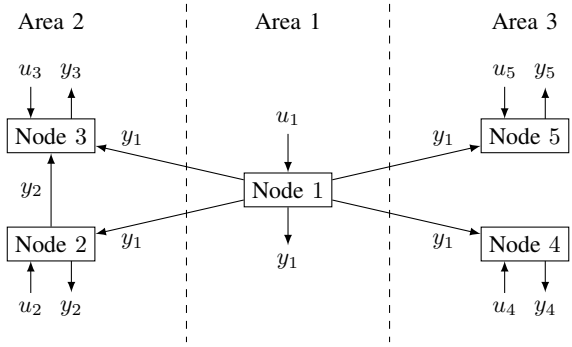
	
	\begin{figure*}
		\begin{equation} \label{eq:u_NRF}
			\small\begin{bmatrix}
				u_1\\u_2\\u_3\\u_4\\u_5
			\end{bmatrix}\hspace{-1mm}=\hspace{-1mm}\tiny\begin{bmatrix}
				0&0&0&0&0\\
				\frac{-0.2}{z-0.8}&0&0&0&0\\
				\frac{-0.2 z + 0.12}{z^2 - 1.6 z + 0.64}&\frac{-0.2}{z-0.8}&0&0&0\\
				\frac{-0.2}{z-0.8}&0&0&0&0\\
				\frac{-0.2}{z-0.8}&0&0&0&0
			\end{bmatrix}\small\hspace{-2mm}\begin{bmatrix}
				u_1\\u_2\\u_3\\u_4\\u_5
			\end{bmatrix}\hspace{-1mm}+\hspace{-1mm}\tiny\begin{bmatrix}
				\frac{1.05z-0.85}{z^2-0.2z-0.8}&0&0&0&0\\
				0&\frac{1.05z-0.85}{z^2-0.2z-0.8}&0&0&0\\
				0&0&\frac{1.05z-0.85}{z^2-0.2z-0.8}&0&0\\
				0&0&0&\frac{1.05z-0.85}{z^2-0.2z-0.8}&0\\
				0&0&0&0&\frac{1.05z-0.85}{z^2-0.2z-0.8}
			\end{bmatrix}\small\hspace{-2mm}\begin{bmatrix}
				z_1\\z_2\\z_3\\z_4\\z_5
			\end{bmatrix}\normalsize\tag{17}
		\end{equation}\hrulefill\vspace{-5mm}
	\end{figure*}

	In this section, we examine the opportunity of implementing NRF-based control schemes as in the SLS framework \cite{matni1,matni2}, by employing the closed-loop maps achievable with stabilizing controllers. We denote by ${\bf T}^{ \ell \epsilon}_{\bf Q}$ the dependency on the Youla parameter ${\bf Q}$ of the closed-loop map between two signals $\epsilon$ to $\ell$, and let $\beta$ denote the states of the controller. We investigate a distributed implementation for controllers based on the closed-loop maps from $[\zeta^\top \; \; r^\top]^\top$ to $[y^\top \; \; u^\top]^\top$ in Fig.~\ref{2Block}, namely
		\begin{align*}
			{\bf T}_{\bf Q}^{y\zeta} \; \beta &= - \; {\bf T}_{\bf Q}^{yr} \; z,  
			\\
			u&=  \phantom{-}\ {\bf T}_{\bf Q}^{u\zeta} \; \beta \; + {\bf T}_{\bf Q}^{ur} \;z, 
		\end{align*}
	or, in terms of the corresponding DCF over $\mathbb{R}_{\mathbb{S}}(\lambda)$, we get that\vspace{-5mm}
	\begin{subequations}
		\begin{align}
			\tilde {\bf Y}_{\bf Q}  \tilde {\bf M}\; \beta  &= \; (\tilde {\bf Y}_{\bf Q}  \tilde {\bf M}-I_p)\; z,  \label{miezu6}\\
			u&= -\tilde {\bf X}_{\bf Q}  \tilde {\bf M}\; \beta \; + \tilde {\bf X}_{\bf Q}  \tilde {\bf M}\;z. \label{miezu5}
		\end{align}
	\end{subequations}

	It can be checked that the elimination of $\beta$ from (\ref{miezu6})-(\ref{miezu5}) (since $\tilde {\bf Y}_{\bf Q}  \tilde {\bf M}$ is invertible, as per Remark \ref{rem:scale}) yields  the ${\bf K_Q}$ controller via its right coprime factorization over $\mathbb{R}_{\mathbb{S}}(\lambda)$ $u= { \bf {\tilde X}_Q} {\tilde {\bf Y}_{\bf Q}}^{-1} \: z$, as in (\ref{YoulaEq}). For implementation purposes, we require an NRF-based formulation of the ``state iteration'' from (\ref{miezu6}). By applying a transformation of type (\ref{NRFmain}), we get
	\begin{equation*}
		\beta= (I_p- {\bf \Omega}^{-1} \tilde {\bf Y}_{\bf Q}  \tilde {\bf M}\; )(\beta+\delta_\beta)  + \; {\bf \Omega}^{-1}(\tilde {\bf Y}_{\bf Q}  \tilde {\bf M}-I_p) z, 
	\end{equation*}		
	where ${\bf \Omega}:= (\tilde {\bf Y}_{\bf Q}  \tilde {\bf M})^{\text{diag}}$, which has a proper inverse (recall Remark \ref{rem:scale}), while $\delta_\beta$ represents additive disturbances acting upon the controller's states, replacing $\delta_u$ from Fig.~\ref{2BlockAgain}. Thus, we obtain a set of control laws which are akin to the \emph{secondary} specialized implementations from Section III-C of \cite{matni2}, namely\vspace{-5mm}
	\begin{subequations}
		\begin{align} \label{redSLS1}
			&\hspace{-2mm}\beta= (I_p- {\bf \Omega}^{-1} \tilde {\bf Y}_{\bf Q}  \tilde {\bf M} )(\beta+\delta_\beta)  + {\bf \Omega}^{-1}(\tilde {\bf Y}_{\bf Q}  \tilde {\bf M}-I_p) z, \\
			\label{redSLS2}
			&\hspace{-2mm}u=  -\tilde {\bf X}_{\bf Q}  \tilde {\bf M}\; \beta \; + \tilde {\bf X}_{\bf Q}  \tilde {\bf M}z.
		\end{align}
	\end{subequations}

	Yet, as expected (see Section III-C of \cite{matni2}), this type of specialized implementation for a centralized stabilizing controller is hampered by restrictive assumptions on the plant's TFM.\vspace{-1mm}
	
	\begin{theorem} \label{MR3}
		
		Let ${\bf G}\in {\mathbb{R}_p(\la)}^{p\times m}$ and consider any of its $\mathbf{Q}$-parameterized DCFs, as in \eqref{EqYoula4}-\eqref{dcrelQ}. Define ${\bf \Omega}^{}:=(\tilde {\bf Y}_{\bf Q}  \tilde {\bf M})^{\text{diag}}$, where ${\bf Q}$ is the (stable) Youla parameter, and let also ${\bf K}_{\mathbf{Q}}\in {\mathbb{R}_p(\la)}^{m\times p}$ be a (centralized) stabilizing controller, as in \eqref{YoulaEq}.
		If ${\bf G}$ is unstable, then the implementation \eqref{redSLS1}-\eqref{redSLS2} of ${\bf K_Q}$ does not internally stabilize the feedback loop with $\mathbf{G}$.\vspace{-1mm}
	\end{theorem}
	
	\begin{proof} See the Appendix.\vspace{-1mm}
	\end{proof}

	\section{Numerical Example}\label{sec:ex}
	
	Consider a grid of $5$ interconnected nodes separated into $3$ local areas, as shown in Fig.~\ref{fig:area}. We aim to obtain a distributed control law in which each node's controller employs only local measurements and exchanges command values only with other sub-controllers that belong to nodes located in the original node's area or in directly adjacent ones.
	
	Thus, we will devise a control law in which the controller of node $1$ sends its command to nodes $2-5$ while the controller of node $2$ sends its command to node $3$.	The network from Fig.~\ref{fig:area} is modeled as a discrete-time system with a sampling time of $T_s=100$ ms. To describe the network's TFM, denoted $\mathbf{G}(z)$, define $\mathbf{\Gamma}_{\mathbf{G}}(z):=\frac{1}{z-1}$ and $\mathbf{\Phi}_{\mathbf{G}}(z):=\frac{0.2}{z-0.8}$ to get that
	\begin{equation*}\label{eq:G_def}
		\begin{array}{c}
			\begin{array}{cc}
				\mathcal{B}:=\footnotesize\begin{bmatrix}
					0&0&0&0&0\\
					1&0&0&0&0\\
					1&1&0&0&0\\
					1&0&0&0&0\\
					1&0&0&0&0\\
				\end{bmatrix},
				& \begin{array}{l}
					\mathbf{U}(z):=I_5-\mathbf{\Phi}_{\mathbf{G}}\mathcal{B},\vspace{2mm}\\
					\mathbf{V}(z):=\mathbf{\Gamma}_{\mathbf{G}}I_5,\vspace{2mm}\\
					\mathbf{G}(z)\phantom{:}=\mathbf{U}^{-1}\mathbf{V},
				\end{array} 
			\end{array}\\
			\mathbf{U}(z)^{-1}=\footnotesize\begin{bmatrix}
				1&0&0&0&0\\
				\mathbf{\Phi}_{\mathbf{G}}&1&0&0&0\\
				\mathbf{\Phi}_{\mathbf{G}}^2+\mathbf{\Phi}_{\mathbf{G}}&\mathbf{\Phi}_{\mathbf{G}}&1&0&0\\
				\mathbf{\Phi}_{\mathbf{G}}&0&0&1&0\\
				\mathbf{\Phi}_{\mathbf{G}}&0&0&0&1\\
			\end{bmatrix}.
		\end{array}
	\end{equation*}
	Note, moreover, that a DCF over $\mathbb{R}_{\mathbb{S}}(\lambda)$ of $\mathbf{G}$ is given by
	\begin{equation*}\label{eq:DCF_V}
		\begin{array}{rcl}
			\begin{bmatrix}
				\phantom{-}\widetilde{\mathbf{M}}&\widetilde{\mathbf{N}}\\
				-{\mathbf{X}}&{\mathbf{Y}}
			\end{bmatrix}&=&\footnotesize\begin{bmatrix}\frac{z-1}{z-0.5}I_5 &\frac{1}{z-0.5}\mathbf{U}^{-1}\vspace{2mm}\\ \frac{-0.25}{z-0.5}I_5&\frac{z}{z-0.5}\mathbf{U}^{-1} \end{bmatrix},\vspace{2mm}\\\begin{bmatrix}
				\widetilde{\mathbf{Y}}&-{\mathbf{N}}\\
				\widetilde{\mathbf{X}}&\phantom{-}{\mathbf{M}}
			\end{bmatrix}&=&\footnotesize\begin{bmatrix} \frac{z}{z-0.5}I_5&\frac{-1}{z-0.5}I_5\vspace{2mm}\\\frac{0.25}{z-0.5}\mathbf{U}&\frac{z-1}{z-0.5}\mathbf{U} \end{bmatrix}.
		\end{array}
	\end{equation*}
	

	\begin{figure}[!t]
		\centering
		\includegraphics[width=.8\columnwidth]{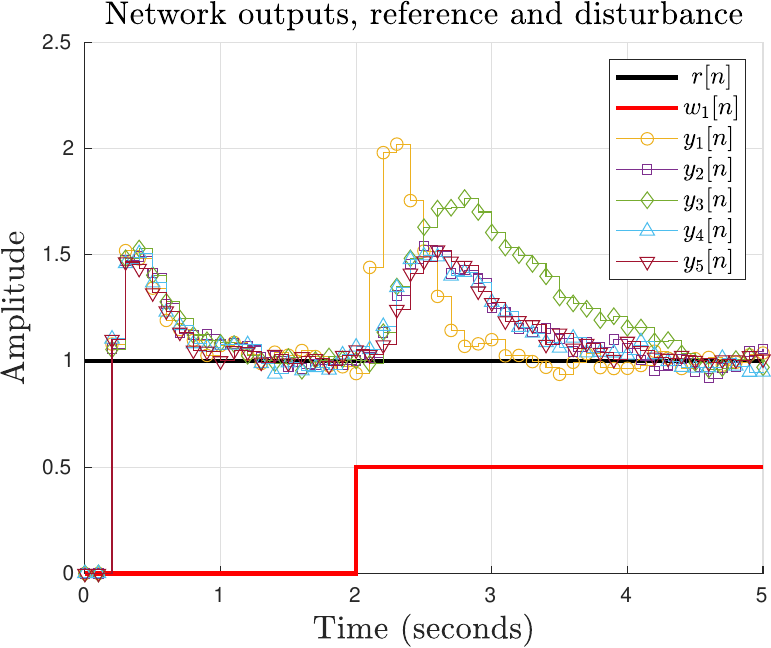}\vspace{-1mm}
		\caption{Reference tracking of the closed-loop network with NRF implementation and input, measurement and communication disturbance\vspace{-6mm}}
		\label{fig:step_resp}
	\end{figure}
	
	\begin{figure*}
		\begin{equation} \label{finallyhappy}
			\ba{c} z \\  u \\ v \\ y\ea = \footnotesize\ba{cccr}  \;- \tilde {\bf Y}_{\bf Q} \tilde {\bf N} \; & \;\tilde {\bf Y}_{\bf Q} \tilde {\bf M} \; &\; -\tilde {\bf Y}_{\bf Q} \tilde {\bf M} \; & \;{\bf N} (\bf Y_Q-\bf Y_Q^{\text{diag}}) \\ \; - \tilde {\bf X}_{\bf Q}  \tilde {\bf N}\; &\; \tilde {\bf X}_{\bf Q}  \tilde {\bf M}\; &\; -\tilde {\bf X}_{\bf Q}  \tilde {\bf M}\; & \; - {\bf M}(\bf Y_Q-\bf Y_Q^{\text{diag}}) \\  \; I_m - \tilde {\bf X}_{\bf Q}  \tilde {\bf N}\; &\; \tilde {\bf X}_{\bf Q}  \tilde {\bf M}\; &\; - \tilde {\bf X}_{\bf Q}  \tilde {\bf M}\; & \; - {\bf M}(\bf Y_Q-\bf Y_Q^{\text{diag}}) \\ \; \tilde {\bf Y}_{\bf Q} \tilde {\bf N} \; & \; I_p-\tilde {\bf Y}_{\bf Q} \tilde {\bf M} \; &\; \phantom{-}\tilde {\bf Y}_{\bf Q} \tilde {\bf M} \; & \;-{\bf N} (\bf Y_Q-\bf Y_Q^{\text{diag}}) \ea   \normalsize\ba{c} w \\ r \\ \zeta \\ \delta_u \ea \tag{21}
		\end{equation}
		\hrulefill
	\end{figure*}
	
	
	We choose $\mathbf{Q}(z)=\frac{0.8}{z-0.2}I_5$ to get, for the network from Fig.~\ref{fig:area}, the control laws from \eqref{eq:u_NRF}, at the top of the next page. 
	
	Consider, now, the following simulation scenario. Let the controller (in NRF form) be implemented in standard unity feedback with our network and let the reference signal be given by $r[n]:=\mathbf{1}[n]\begin{bmatrix}
		1&1&1&1&1
	\end{bmatrix}^\top$, where $\mathbf{1}[n]$ denotes the discrete-time Heaviside step function.

Also, let each output measurement be disturbed additively by a measurement noise $\zeta_i[n]$, $i\in1:5$, and let each communicated command be affected additively by a communication disturbance $\delta_{ui}[n]$, $i\in1:5$, with these signals being modeled as uniformly distributed noise having $\left|\zeta_i[n]\right|,\left|\delta_{ui}[n]\right|\leq0.05,\ \forall n\in\mathbb{N}$. Moreover, let $w_1[n]:=0.5\times\mathbf{1}[n-20]$ be an additive disturbance at the input of the network's first node. 
	
	As can be seen in Fig.~\ref{fig:step_resp} on the next page, not only are all signals bounded, even in the presence of communication disturbance, but the distributed controller also ensures satisfactory performance for reference tracking and disturbance rejection.\vspace{-1mm}
	
	\section{Conclusion}\label{sec:fin}
	
	We have successfully formalized the distributed design framework introduced in \cite{TAC2016}, while highlighting its inherent connections \cite{Luca3} with existing techniques from literature \cite{matni1,matni2}. Moreover, we have extended the coprime factor approach from \cite{SI} via state-space-based implementations, while laying the theoretic groundwork for reliable numerical procedures \cite{aug_sparse}, enabling the synthesis of scalable and robust control laws.\vspace{-2mm}
	
	\section*{Acknowledgment}
	
	The authors would like to thank our colleague, Dr. Bogdan D. Ciubotaru, for the insightful comments and advice made during the elaboration of this manuscript.\vspace{-2mm}
	

	\section*{Appendix}\vspace{-1mm}
	
	\noindent {\bf Proof of Theorem~\ref{IS}} In order to prove point $\mathbf{(a)}$ from the result's statement, we will express all closed-loop maps from $\begin{bmatrix}
		r^\top & w ^\top & \zeta^\top & \delta_u^\top \end{bmatrix}^\top$ to $\begin{bmatrix}
		y^\top & u^\top & z^\top & v^\top \end{bmatrix}^\top$ in terms of the $\mathbf{Q}$-parameterized DCF given in \eqref{EqYoula4}-\eqref{dcrelQ}. 
	The equations of the standard unity feedback interconnection from Fig.~\ref{2Block} are given by $z  =r - y $, $v  =u+w$, $y  = {\bf G} v + \zeta$ and $u  = {\bf K}_{\bf Q}\: z$, respectively,
	or equivalently by $y  =r -z $, $v   =u + w $, and\stepcounter{equation}\vspace{-1mm}
	\begin{subequations}
		\begin{align}
			z + {\bf G}\: u  &= -{\bf G} w + r - \zeta \label{pl1} \: ,\\
			- {\bf K}_{\bf Q} z +u &=O  \label{pl2}.\vspace{-1mm}
		\end{align}
	\end{subequations}
	Multiplying to the left in \eqref{pl1} with $\tilde {\bf M}$  and in \eqref{pl2} with ${\bf Y}_{\bf Q}$, we obtain, via the identities from \eqref{EqYoula4}-\eqref{dcrelQ}, that 
	\begin{subequations}\vspace{-1mm}
		\begin{align} 
			\tilde {\bf M} z + \tilde {\bf N} u   &= - \tilde  {\bf N} w + \tilde  {\bf M}   r  - \tilde  {\bf M}   \zeta \: , \label{temp0} \\ 
			- {\bf X}_{\bf Q} z  +{\bf Y}_{\bf Q} u &=O . \label{temp1}\vspace{-1mm}
		\end{align}
	\end{subequations}
	By  implementing $u$ via \eqref{impletrue} and \eqref{ISmain1}-\eqref{ISmain2} and allowing it to be affected by $\delta_u$, 
	the equations from (\ref{temp0})-(\ref{temp1}) turn into
	\begin{align}\nonumber 
		&\footnotesize\ba{cc}   \tilde {\bf M} &  \tilde {\bf N} \\ - \big( \bf Y_Q^{\text{diag}}\big)^{-1} {\bf X_Q} \; &  \; {\big( \bf Y_Q^{\text{diag}}\big)}^{-1} {\bf Y}_{\bf Q} \ea \ba{c} z \\  u \ea = \\ &\footnotesize\ba{cccc} -\tilde {\bf N} &  \tilde {\bf M} \; & -\tilde {\bf M} \; & O \\  O\; &  O & O &  \; I_m - {\big( \bf Y_Q^{\text{diag}}\big)}^{-1} {\bf Y}_{\bf Q} \ea \tiny\ba{c} w \\ r \\ \zeta \\ \delta_u \ea.\normalsize\label{temp33}
	\end{align}
	
	To obtain the explicit dependency of $\begin{bmatrix}
		z^\top & u^\top \end{bmatrix}^\top$ in terms of $\begin{bmatrix}
		r^\top & w ^\top & \zeta^\top & \delta_u^\top \end{bmatrix}^\top$, we multiply \eqref{temp33} to the left with 
	\begin{equation*} \label{dcrelmodbis}
		\footnotesize\ba{cc}   \tilde {\bf M} &  \tilde {\bf N} \\ - \big( \bf Y_Q^{\text{diag}}\big)^{-1} {\bf X_Q} &\hspace{-2mm}  \big( \bf Y_Q^{\text{diag}}\big)^{-1} {\bf Y_Q} \ea ^{-1}\hspace{-2mm}=\footnotesize\ba{cr}  \tilde {\bf Y}_{\bf Q} &\hspace{-2mm} -{\bf N} \: \bf Y_Q^{\text{diag}}\\   \tilde {\bf X}_{\bf Q} &  {\bf M} \: \bf Y_Q^{\text{diag}} \ea.
		\normalsize
	\end{equation*}
	Moreover, we have from Fig.~\ref{2BlockAgain} that $v   =u + w$ and that $y  =r -z $. Then, the resulting closed-loop maps will be given by (\ref{finallyhappy}), at the top of the next page. Since $\mathbf{Q}$ is stable, all of them will also be stable, thus guaranteeing internal stability.
	\qed
	\medskip

	\noindent {\bf Proof of Lemma \ref{lem:pole_conserv}} 
	We focus on the case where $\mathbf{G}_2$ is not a constant matrix and $\mathcal{P}_u(\mathbf{G}_1)\neq\{\emptyset\}$ since, otherwise, the result follows directly from the full row rank of $\mathbf{G}_2$, via classical Popov-Belevitch-Hautus (PBH) tests (see Section 3.2 of \cite{zhou}), or from the stability of $\mathbf{G}_1$. We start by expressing $\mathbf{G}_1$ and $\mathbf{G}_2$ via the minimal realizations\stepcounter{equation}
	\begin{equation}\label{eq:min1}
			\mathbf{G}_i(\lambda)=\left[\footnotesize\begin{array}{c|c}
				A_i-\lambda I_{n_i}&B_i\\\hline C_i&D_i
			\end{array}\right],\ i\in\{1,2\},
	\end{equation}
	and we use these in order to write down the realization of
	\begin{equation}\label{eq:min2}
		\mathbf{G}_1(\lambda)\mathbf{G}_2(\lambda)=\left[\footnotesize\begin{array}{cc|c}
			A_1-\lambda I_{n_1}&B_1C_2&B_1D_2\\O&A_2-\lambda I_{n_2}&B_2\\\hline C_1&D_1C_2&D_1D_2
		\end{array}\right],
	\end{equation}
	which we will show to be both stabilizable and detectable.
	
	Since $\mathbf{G}_2$ has full row normal rank and no zeros in $\mathbb{C}\backslash\mathbb{S}$, then by Lemma 3.33 and Theorem 3.34 in \cite{zhou} and from the minimality of \eqref{eq:min1}, we get that $\mathbf{S}(\lambda):=\footnotesize\begin{bmatrix}
		A_2-\lambda I_{n_2}&B_2\\C_2&D_2
	\end{bmatrix}$ has full row rank $\forall\lambda\in\mathbb{C}\backslash\mathbb{S}$. From this and the stabilizability of the pair $(A_1,B_1)$, recalling the minimality of \eqref{eq:min1}, a standard PBH test confirms that \eqref{eq:min2} is also stabilizable.
	
	From the minimality of \eqref{eq:min1}, we have that the pair $(C_1,A_1)$ is detectable and that $\mathcal{P}_u(\mathbf{G}_2)=\{\emptyset\}\Rightarrow\Lambda_u(A_2)=\{\emptyset\}$. By employing a PBH test, we get that \eqref{eq:min2} is also detectable, in addition to being stabilizable, and therefore $\mathcal{P}_u(\mathbf{G}_1\mathbf{G}_2)=\Lambda_u(A_1)\cup\Lambda_u(A_2)$. Yet, from the minimality of \eqref{eq:min1} and by recalling the stability of $\mathbf{G}_2$, we have that $\mathcal{P}_u(\mathbf{G}_1)=\Lambda_u(A_1)$ and that $\Lambda_u(A_2)=\{\emptyset\}$. Thus, $\mathcal{P}_u(\mathbf{G}_1\mathbf{G}_2)=\mathcal{P}_u(\mathbf{G}_1)$.\qed\medskip

	\noindent {\bf Proof of Theorem \ref{thm:implem}} The proof can be broken down into four parts. In the first part, $\mathbf{(I)}$, we prove that $\mathcal{P}_u\left(\begin{bmatrix}
		\mathbf{\Phi}&\mathbf{\Gamma}
	\end{bmatrix}\right)=\textstyle\bigcup_{i=1}^m\mathcal{P}_u(e_i^\top\begin{bmatrix}
		\mathbf{\Phi}&\mathbf{\Gamma}
	\end{bmatrix})$. By employing this fact along with the stabilizable and detectable realizations denoted
	\begin{equation}\label{eq:row_real}
		e_i^\top\begin{bmatrix}
			\mathbf{\Phi} & \mathbf{\Gamma}
		\end{bmatrix}=\left[\footnotesize\begin{array}{c|c}
			A_i-\lambda I_{n_i}& B_i\\\hline C_i&D_i
		\end{array}\right],\ \forall i\in1:m,
	\end{equation}
	we prove in part $\mathbf{(II)}$ the fact that the resulting realization of the controller's row-based NRF implementation, namely
	\begin{equation}\label{eq:full_real}
		\scriptsize\begin{bmatrix}
			e_1^\top\begin{bmatrix}
				\mathbf{\Phi} & \mathbf{\Gamma}
			\end{bmatrix}\\
			\vdots\\
			e_i^\top\begin{bmatrix}
				\mathbf{\Phi} & \mathbf{\Gamma}
			\end{bmatrix}\\
			\vdots\\
			e_m^\top\begin{bmatrix}
				\mathbf{\Phi} & \mathbf{\Gamma}
			\end{bmatrix}
		\end{bmatrix}\hspace{-1mm}=\hspace{-1mm}\left[\scriptsize\begin{array}{ccc|c}
			A_1-\lambda I_{n_1}&&& B_1\\
			&\hspace{-2mm}\ddots&&\vdots\\
			&&\hspace{-2mm}A_m-\lambda I_{n_m}&B_m\\
			\hline C_1&&&D_1\\
			&\hspace{-2mm}\ddots&&\vdots\\
			&&\hspace{-2mm}C_m&D_m
		\end{array}\right],\normalsize
	\end{equation}
	is both stabilizable and detectable. In part $\mathbf{(III)}$, we show that the NRF implementation solves a more general stabilization problem and, in part $\mathbf{(IV)}$, we employ parts $\mathbf{(II)}$ and $\mathbf{(III)}$ to prove that these state-space implementations ensure that the closed-loop system's state dynamics are asymptotically stable.

	$\mathbf{(I)}$ Notice that, since $\mathbf{Y}_{\mathbf{Q}}^{\text{diag}}$ is a diagonal TFM, we have \vspace{-1mm}
	$$\begin{array}{rcl}
		\begin{bmatrix}
			\mathbf{\Phi}&\mathbf{\Gamma}
		\end{bmatrix}&\hspace{-2mm}=&\hspace{-2mm}\begin{bmatrix}
			I_m&O
		\end{bmatrix}-(\mathbf{Y}_{\mathbf{Q}}^{\text{diag}})^{-1}\begin{bmatrix}
			\mathbf{Y}_{\mathbf{Q}} & -\mathbf{X}_{\mathbf{Q}} 
		\end{bmatrix},\\e_i^\top\begin{bmatrix}
			\mathbf{\Phi}&\mathbf{\Gamma}
		\end{bmatrix}&\hspace{-2mm}=&\hspace{-2mm} \begin{bmatrix}
			e_i^\top&O
		\end{bmatrix}-e_i^\top(\mathbf{Y}_{\mathbf{Q}}^{\text{diag}})^{-1}e_ie_i^\top\begin{bmatrix}
			\mathbf{Y}_{\mathbf{Q}} & -\mathbf{X}_{\mathbf{Q}} 
		\end{bmatrix}.\end{array}\vspace{-1mm}$$ 
	Since $\begin{bmatrix}
		-\mathbf{X}_{\mathbf{Q}} & \mathbf{Y}_{\mathbf{Q}} 
	\end{bmatrix}$ is stable and since it satisfies \eqref{dcrelQ}, then $\begin{bmatrix}
		-\mathbf{X}_{\mathbf{Q}}(\lambda) & \mathbf{Y}_{\mathbf{Q}}(\lambda) 
	\end{bmatrix}$ has only finite entries along with full row rank $\forall\lambda\in\mathbb{C}\backslash\mathbb{S}$. This means that $\begin{bmatrix}
		\mathbf{Y}_{\mathbf{Q}} & -\mathbf{X}_{\mathbf{Q}} 
	\end{bmatrix}$ and every $e_i^\top\begin{bmatrix}
		\mathbf{Y}_{\mathbf{Q}} & -\mathbf{X}_{\mathbf{Q}} 
	\end{bmatrix}$ $\forall i\in1:m$ will share this property. Then, it follows that $\begin{bmatrix}
		\mathbf{Y}_{\mathbf{Q}} & -\mathbf{X}_{\mathbf{Q}} 
	\end{bmatrix}$ and every $e_i^\top\begin{bmatrix}
		\mathbf{Y}_{\mathbf{Q}} & -\mathbf{X}_{\mathbf{Q}} 
	\end{bmatrix}$ $\forall i\in1:m$ are stable and that they have full row normal rank along with, by Lemma 3.29 in \cite{zhou}, no transmission zeros in $\mathbb{C}\backslash\mathbb{S}$.
	
	We now apply Lemma \ref{lem:pole_conserv} to both $(\mathbf{Y}_{\mathbf{Q}}^{\text{diag}})^{-1}\begin{bmatrix}
		\mathbf{Y}_{\mathbf{Q}} & -\mathbf{X}_{\mathbf{Q}} 
	\end{bmatrix}$ and $e_i^\top(\mathbf{Y}_{\mathbf{Q}}^{\text{diag}})^{-1}e_ie_i^\top\begin{bmatrix}
		\mathbf{Y}_{\mathbf{Q}} & -\mathbf{X}_{\mathbf{Q}} 
	\end{bmatrix}$ in order to get that\vspace{-1mm}\footnotesize
	$$\begin{array}{rcl}
	\mathcal{P}_u\left(\begin{bmatrix}
		\mathbf{\Phi}&\hspace{-1mm}\mathbf{\Gamma}
	\end{bmatrix}\right)=\mathcal{P}_u(\begin{bmatrix}
			I_m&\hspace{-1mm}O
		\end{bmatrix}-\begin{bmatrix}
			\mathbf{\Phi}&\hspace{-1mm}\mathbf{\Gamma}
		\end{bmatrix})&\hspace{-3mm}=&\hspace{-3mm}\mathcal{P}_u((\mathbf{Y}_{\mathbf{Q}}^{\text{diag}})^{-1}),\\
		\mathcal{P}_u(e_i^\top\begin{bmatrix}
			\mathbf{\Phi}&\hspace{-1mm}\mathbf{\Gamma}
		\end{bmatrix})=\mathcal{P}_u(\begin{bmatrix}
			e_i^\top&\hspace{-1mm}O
		\end{bmatrix}-e_i^\top\begin{bmatrix}
			\mathbf{\Phi}&\hspace{-1mm}\mathbf{\Gamma}
		\end{bmatrix})&\hspace{-3mm}=&\hspace{-3mm}\mathcal{P}_u(e_i^\top(\mathbf{Y}_{\mathbf{Q}}^{\text{diag}})^{-1}e_i).
	\end{array}\vspace{-1mm}$$\normalsize
	Moreover, from the diagonal structure of $\mathbf{Y}_{\mathbf{Q}}^{\text{diag}}$, it is straightforward to obtain, using classical state-space theory, that $\mathcal{P}_u((\mathbf{Y}_{\mathbf{Q}}^{\text{diag}})^{-1})=\textstyle\bigcup_{i=1}^m\mathcal{P}_u(e_i^\top(\mathbf{Y}_{\mathbf{Q}}^{\text{diag}})^{-1}e_i)$. This, in turn, yields the fact that
	$\mathcal{P}_u\left(\begin{bmatrix}
		\mathbf{\Phi}&\mathbf{\Gamma}
	\end{bmatrix}\right)=\textstyle\bigcup_{i=1}^m\mathcal{P}_u(e_i^\top\begin{bmatrix}
		\mathbf{\Phi}&\mathbf{\Gamma}
	\end{bmatrix}).$

	$\mathbf{(II)}$ We now turn to the stabilizable and detectable realizations for each $e_i^\top\begin{bmatrix}
		\mathbf{\Phi} & \mathbf{\Gamma}
	\end{bmatrix}$ given in \eqref{eq:row_real}. By the stabilizability and detectability of these realizations, we have that $\mathcal{P}_u(e_i^\top\begin{bmatrix}
		\mathbf{\Phi}&\mathbf{\Gamma}
	\end{bmatrix})=\Lambda_u(A_i)$. Then, we employ these realizations to form the one from \eqref{eq:full_real} and we define $A_{\mathbf{K}}:=\text{diag}(A_1,\cdots,A_m)$,
	which is precisely the state matrix of the realization from \eqref{eq:full_real}. By employing the fact that $A_{\mathbf{K}}$ is block-diagonal, we get that $\Lambda_u(A_{\mathbf{K}})=\textstyle\bigcup_{i=1}^m\Lambda_u(A_i)$, which implies $\Lambda_u(A_{\mathbf{K}})=\textstyle\bigcup_{i=1}^m\mathcal{P}_u(e_i^\top\begin{bmatrix}
		\mathbf{\Phi}&\mathbf{\Gamma}
	\end{bmatrix})=\mathcal{P}_u\left(\begin{bmatrix}
		\mathbf{\Phi}&\mathbf{\Gamma}
	\end{bmatrix}\right).$
	Therefore, it follows (by standard state-space theory) that the realization from \eqref{eq:full_real} must be both detectable and stabilizable.

	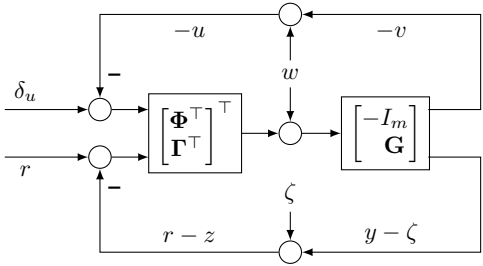
\begin{figure}
		\centering
		\begin{tikzpicture}[scale=.9, every node/.style={transform shape}]
			\node(n1)at(0,0){};
			\node(n2)[right of=n1, node distance=4em]{};
			\node(n21)[circle,draw,minimum height=8,above of=n2, node distance=1em]{};
			\node(n22)[circle,draw,minimum height=8,below of=n2, node distance=1em]{};
			\node(n3)[minimum height=0.5cm, minimum width=0.5cm,draw,right of=n2, node distance=4em]{$ \begin{bmatrix}
					\mathbf{\Phi}^\top\\\mathbf{\Gamma}^\top
				\end{bmatrix}^\top $};
			\node(n4)[circle,draw,minimum height=8,right of=n3, node distance=4em]{};
			\node(n41)[circle,draw,minimum height=8,above of=n4, node distance=5em]{};
			\node(n42)[circle,draw,minimum height=8,below of=n4, node distance=5em]{};
			\node(n42m)[below of=n4, node distance=2.5em]{$\zeta$};
			\node(n41m)[above of=n4, node distance=2.5em]{$w$};
			\node(n5)[minimum height=0.5cm, minimum width=0.5cm,draw,right of=n4, node distance=4em]{$\begin{bmatrix}
					-I_m\\\phantom{-}\mathbf{G}
				\end{bmatrix}$};
			
			\draw[-latex](n41m)--(n4) node[near end, right]{$ $};
			\draw[-latex](n42m)--(n42) node[near end, left]{$ $};
			\draw[-latex](n41m)--(n41) node[near end, left]{$ $};
			
			\draw[-latex](n3)--(n4) node[near end, above]{$ $};
			\draw[-latex](n4)--(n5);
			
			\draw[-latex](17.8em,1em)--(20em,1em)|-(n41) node[near end, below]{$-v$} node[very near end, below]{\hspace{-2em}$ $};
			\draw[-latex](17.8em,-1em)--(20em,-1em)|-(n42) node[near end, above]{$y-\zeta$} node[very near end, above]{\hspace{-2em}$ $};
			
			\draw[-latex](n41)-|(n21) node[near start, below]{$-u$} node[very near end, right]{\bf{--}};
			\draw[-latex](n42)-|(n22) node[near start, above]{$r-z$} node[very near end, right]{\bf{--}};
			
			\draw[-latex](n21)--(6.05em,1em);
			\draw[-latex](n22)--(6.05em,-1em);
			
			\draw[-latex](0em,1em)--(n21) node[near start,above]{$\delta_u$} node[near end, above]{$ $};
			\draw[-latex](0em,-1em)--(n22) node[near start,below]{$r$} node[near end, below]{$ $};
		\end{tikzpicture}\vspace{-1mm}
		\caption{Equivalent negative unity feedback interconnection\vspace{-5mm}}\label{fig:eq_con}
	\end{figure}
\begin{figure*}
	\begin{equation}\label{eq:spaghet}
		\begin{bmatrix}
			\sigma x_{\mathbf{G}}\\\sigma x_{\mathbf{K}}
		\end{bmatrix}\hspace{-1mm}=A_{CL}\begin{bmatrix}
			x_{\mathbf{G}}\\ x_{\mathbf{K}}
		\end{bmatrix}\hspace{-1mm}+\hspace{-1mm}\footnotesize\begin{bmatrix}
			O&O&B\\-B_{\mathbf{K}1}&-B_{\mathbf{K}2}&O
		\end{bmatrix}\footnotesize\hspace{-1mm}\widetilde{D}^{-1}\hspace{-1mm}\begin{bmatrix}
			O&I_m&O&-I_m&O&O\\
			-I_p&O&I_p&O&I_p&O\\
			O&I_m&O&O&O&I_m
		\end{bmatrix}\begin{bmatrix}
			r^\top & w^\top & \zeta^\top & \delta_u^\top & \delta_{y_{\mathbf{G}}}^\top & \delta_{y_{\mathbf{K}}}^\top \end{bmatrix}^\top\hspace{-1mm}\normalsize+\begin{bmatrix}
			\delta_{x_{\mathbf{G}}} \\ \delta_{x_{\mathbf{K}}}
		\end{bmatrix}\tag{32}
		\normalsize\end{equation}\hrulefill
\end{figure*}
	
	$\mathbf{(III)}$ We now show that $\begin{bmatrix}
		\mathbf{\Phi}&\mathbf{\Gamma}
	\end{bmatrix}$ internally stabilizes $\begin{bmatrix}
		-I_m&\mathbf{G}^\top
	\end{bmatrix}\hspace{-1mm}\phantom{.}^\top$ in standard unity configuration (recall Fig.~\ref{2Block}). Moreover, we point out that this fact is a sufficient condition for the feedback configuration from Fig.~\ref{fig:eq_con}, which is equivalent to the one in Fig.~\ref{2BlockAgain}, to be internally stable.
	
	Note that $\begin{bmatrix}
		\mathbf{\Phi}&\mathbf{\Gamma}
	\end{bmatrix}$ internally stabilizes $\begin{bmatrix}
		-I_m&\mathbf{G}^\top
	\end{bmatrix}\hspace{-1mm}\phantom{.}^\top$ if and only if (see Lemma 5.3 of \cite{zhou} for the continuous-time positive feedback case) all the entries of the following TFM
	\begin{equation}\label{eq:Ht}
		\widetilde{\mathbf{H}}_{CL}:=\footnotesize\begin{bmatrix}
			\phantom{-}I_m\\-I_m\\\phantom{-}\mathbf{G}
		\end{bmatrix}\left(I_m+\begin{bmatrix}
			\mathbf{\Phi}&\mathbf{\Gamma}
		\end{bmatrix}\begin{bmatrix}
			-I_m\\\phantom{-}\mathbf{G}
		\end{bmatrix}\right)^{-1}\begin{bmatrix}
			I_m&\mathbf{\Phi}&\mathbf{\Gamma}
		\end{bmatrix}\normalsize
	\end{equation}
	are stable. Recall \eqref{EqYoula4}-\eqref{YoulaEq} along with \eqref{ISmain1}-\eqref{ISmain2} to get that 
	\begin{equation*}
		\widetilde{\mathbf{H}}_{CL}=
			\footnotesize\begin{bmatrix}
				\mathbf{M}^\top&-\mathbf{M}^\top&\mathbf{N}^\top
			\end{bmatrix}^\top\begin{bmatrix}
				\mathbf{Y}_{\mathbf{Q}}^{{\text{diag}}}&\mathbf{Y}_{\mathbf{Q}}^{{\text{diag}}}-\mathbf{Y}_{\mathbf{Q}}&\mathbf{X}_{\mathbf{Q}}
			\end{bmatrix}
		\end{equation*}
		has only stable entries. Therefore, $\begin{bmatrix}
			\mathbf{\Phi}&\mathbf{\Gamma}
		\end{bmatrix}$ internally stabilizes $\begin{bmatrix}
			-I_m&\mathbf{G}^\top
		\end{bmatrix}\hspace{-1mm}\phantom{.}^\top$ which also confirms, recalling \eqref{finallyhappy} and the equivalence with Fig.~\ref{2BlockAgain}, that Fig.~\ref{fig:eq_con} is indeed internally stable.
		
		$\mathbf{(IV)}$ Finally, recall from the result's statement that the plant is described by a stabilizable and detectable realization
		\begin{subequations}
			\begin{align} 
				\sigma x_\mathbf{G}  &=  A x_\mathbf{G}  +  Bu_\mathbf{G},\label{eq:plt_ss1} \\ 
				y_\mathbf{G} &=  C  x_\mathbf{G}  + D u_\mathbf{G},\label{eq:plt_ss2}
			\end{align}
		\end{subequations}
		and then denote the realization of $\begin{bmatrix}
			\mathbf{\Phi}&\mathbf{\Gamma}
		\end{bmatrix}$ from \eqref{eq:full_real} as follows\vspace{-5mm}
		\begin{subequations}
			\begin{align} 
				\sigma x_\mathbf{K}  &=  A_{\mathbf{K}} x_{\mathbf{K}}  +  B_1u_{\mathbf{K}1}+B_2u_{\mathbf{K}2},\label{eq:ctl_ss1} \\ 
				y_{\mathbf{K}} &=  C_{\mathbf{K}}  x_{\mathbf{K}}  + D_1 u_{\mathbf{K}1}+D_2u_{\mathbf{K}2}.\label{eq:ctl_ss2}
			\end{align}
		\end{subequations}
		With these representations, it is straightforward to check by direct substitution that the closed-loop state dynamics of type \eqref{ss0ab}-\eqref{ss0c} in Fig.~\ref{2BlockAgain} are described by the same set of equations (with respect to the same sets of inputs and outputs) as those in Fig.~\ref{fig:eq_con} when $\begin{bmatrix}
			-I_m&\mathbf{G}^\top
		\end{bmatrix}\hspace{-1mm}\phantom{.}^\top$ is described by the realization\vspace{-5mm}
		\begin{subequations}
			\begin{align} 
				\sigma x_\mathbf{G}  &=  A x_\mathbf{G}  +  Bu_\mathbf{G},\label{eq:plt_ex_ss1} \\ 
				-u_\mathbf{G} &=  O  x_\mathbf{G}  - I_m u_\mathbf{G},\label{eq:plt_ex_ss2}\\
				y_\mathbf{G} &=  C  x_\mathbf{G}  + D u_\mathbf{G}.\label{eq:plt_ex_ss3}
			\end{align}
		\end{subequations}
		Note that, since \eqref{eq:plt_ss1}-\eqref{eq:plt_ss2} is stabilizable and detectable, then so is \eqref{eq:plt_ex_ss1}-\eqref{eq:plt_ex_ss3} and that these two realizations share the same state variables, \emph{i.e.}, the components of the plant's state vector.
		
		Consider the state vector of the closed-loop interconnections in both Fig.~\ref{2BlockAgain} and in Fig.~\ref{fig:eq_con} to be the concatenation of the plant's state vector, $x_\mathbf{G}$, with the distributed controller's state vector, $x_{\mathbf{K}}$. Following this, define
		$\widetilde{D}:=\tiny\begin{bmatrix}
			I_m&O&\phantom{-}I_m\\
			O&I_p&-D\\
			D_{\mathbf{K}1}&D_{\mathbf{K}2}&\phantom{-}I_m
		\end{bmatrix}.$
		To show that $\widetilde{D}$ is invertible, we compute the Schur complement of its upper left $(m+p)\times(m+p)$ block and obtain
		\begin{multline}\label{eq:SC}
			I_m+\begin{bmatrix}
				D_{\mathbf{K}1}&\hspace{-1mm}D_{\mathbf{K}2}
			\end{bmatrix}\begin{bmatrix}
				-I_m&\hspace{-1mm}D^\top
			\end{bmatrix}\hspace{-1mm}\phantom{.}^\top=\\=I_m+\begin{bmatrix}
				\mathbf{\Phi}(\infty)&\hspace{-1mm}\mathbf{\Gamma}(\infty)
			\end{bmatrix}\begin{bmatrix}
				-I_m&\hspace{-1mm}\mathbf{G}^\top(\infty)
			\end{bmatrix}\hspace{-1mm}\phantom{.}^\top.
		\end{multline}
		Since $I_m+\begin{bmatrix}
			\mathbf{\Phi}&\mathbf{\Gamma}
		\end{bmatrix}\begin{bmatrix}
			-I_m&\mathbf{G}^\top
		\end{bmatrix}^\top$ is proper and its inverse is, by point $\mathbf{(III)}$, both proper and stable, then the Schur complement from \eqref{eq:SC} is invertible, implying 
		that $\widetilde{D}$ is also invertible. We now combine \eqref{eq:plt_ss1}-\eqref{eq:plt_ss2} and \eqref{eq:plt_ex_ss1}-\eqref{eq:plt_ex_ss3} with \eqref{eq:ctl_ss1}-\eqref{eq:ctl_ss2} to get the fact that the closed-loop interconnection's realization in both Fig.~\ref{2BlockAgain} and in Fig.~\ref{fig:eq_con} has the following state matrix\vspace{-1mm}
		\begin{equation}
			A_{CL}=\footnotesize\begin{bmatrix}
				A&O\\O&A_{\mathbf{K}}
			\end{bmatrix}+\footnotesize\begin{bmatrix}
				O&O&B\\-B_{\mathbf{K}1}&-B_{\mathbf{K}2}&O
			\end{bmatrix}\hspace{-1mm}\footnotesize\widetilde{D}^{-1}\hspace{-1mm}\footnotesize\begin{bmatrix}
				O&O\\C&O\\O&C_{\mathbf{K}}
			\end{bmatrix}\hspace{-1mm}.\label{eq:A_CL}\normalsize
		\end{equation}
		
		Since the TFM from \eqref{eq:Ht} is stable, while \eqref{eq:plt_ex_ss1}-\eqref{eq:plt_ex_ss3} and \eqref{eq:ctl_ss1}-\eqref{eq:ctl_ss2} are stabilizable and detectable, we apply the negative unity feedback versions of Lemmas 5.2 and 5.3 from \cite{zhou} to the feedback loop from Fig.~\ref{fig:eq_con} (directly in the continuous-time case and in adapted form for discrete-time) to get that $\Lambda_u(A_{CL})=\{\emptyset\}$. This is equivalent, by Definition 5.2 in \cite{zhou} (with the appropriate alteration for the discrete-time case), to the desired result, \emph{i.e.}, the two state vectors, $x_{\mathbf{G}}$ and $x_{\mathbf{K}}$, in Fig.~\ref{2BlockAgain} or in Fig.~\ref{fig:eq_con} are driven asymptotically to the zero vector, when evolving freely from any finite initial conditions.		\qed\medskip

		\noindent {\bf Proof of Corollary~\ref{cor:delta}} We start by recalling the realizations from \eqref{eq:plt_ss1}-\eqref{eq:plt_ss2} and \eqref{eq:ctl_ss1}-\eqref{eq:ctl_ss2}, along with the matrix $A_{CL}$ from \eqref{eq:A_CL}. As discussed in part $\mathbf{(IV)}$ of the proof belonging to Theorem \ref{thm:implem}, the closed-loop state dynamics of type \eqref{ss0ab} in Fig.~\ref{2BlockAgain}, when considering the closed-loop state variables to be the concatenation of $x_{\mathbf{G}}$ (the state vector of $\mathbf{G}$) and of $x_{\mathbf{K}}$ $\big($the state vector of $\begin{bmatrix}
			\mathbf{\Phi}&\mathbf{\Gamma}
		\end{bmatrix}$$\big)$, are given by \eqref{eq:spaghet} at the top of this page. Since $\Lambda_u(A_{CL})=\{\emptyset\}$, the closed-loop system in Fig.~\ref{2BlockAgain} is guaranteed to be internally stable, even when taking into account the bounded disturbances from the statement.\qed\medskip

		\noindent {\bf Proof of Theorem~\ref{MR3}} The proof boils down to showing the fact that the closed-loop map from the input disturbance $w$ to the controller's state $\beta$, denoted $\mathbf{T}_{\mathbf{Q}}^{\beta w}$, is unstable if so is $\mathbf{G}$.
		
		Begin with the closed-loop equations, which are given by\stepcounter{equation}
		\begin{subequations}
			\begin{align}
				&\tilde {\bf M} \: z + \tilde {\bf N} \:u   = - \tilde  {\bf N} \:w + \tilde  {\bf M}  \: r  - \tilde  {\bf M}  \: \zeta,  \label{miezu10}\\
				&{\bf \Omega}^{-1}(I_p-\tilde {\bf Y}_{\bf Q}  \tilde {\bf M})z + {\bf \Omega}^{-1} \tilde {\bf Y}_{\bf Q}  \tilde {\bf M} \beta  =  (I_p- {\bf \Omega}^{-1} \tilde {\bf Y}_{\bf Q}  \tilde {\bf M} )\delta_\beta,  \label{miezu11}\\
				&\tilde {\bf X}_{\bf Q}  \tilde {\bf M}\;z -  \tilde {\bf X}_{\bf Q}  \tilde {\bf M}\; \beta -u =  \; O, \label{miezu12}
			\end{align}
		\end{subequations}\normalsize
		where (\ref{miezu11})-(\ref{miezu12}) represent the distributed implementation of the controller. Next, multiply (\ref{miezu10}) to the left with $\tilde {\bf Y}_{\bf Q}$ and rewrite (\ref{miezu10})-(\ref{miezu12}) in matrix form in order to get that
		\begin{multline} \label{cutata}
			\hspace{-2mm}\footnotesize\ba{ccc} \tilde {\bf Y}_{\bf Q} \tilde {\bf M} &  O &   \tilde {\bf Y}_{\bf Q}  \tilde {\bf N} \\
			{\bf \Omega}^{-1}(I_p-\tilde {\bf Y}_{\bf Q}  \tilde {\bf M}) & {\bf \Omega}^{-1} \tilde {\bf Y}_{\bf Q}  \tilde {\bf M} & O \\
			\tilde {\bf X}_{\bf Q}  \tilde {\bf M} & -  \tilde {\bf X}_{\bf Q}\tilde {\bf M} & -I_m \ea \small\ba {c} z \\ \beta \\ u \ea =\\ 
			\footnotesize\ba{cccc} - \tilde {\bf Y}_{\bf Q} \tilde {\bf N} & \tilde {\bf Y}_{\bf Q} \tilde {\bf M}  & - \tilde  {\bf Y}_{\bf Q}  \tilde {\bf M} & O \\
			O & O & O & I_p- {\bf \Omega}^{-1}\tilde {\bf Y}_{\bf Q}  \tilde {\bf M}\\
			O & O & O & O \
			\ea
			\scriptsize\ba{c} w \\ r \\ \zeta \\ \delta_\beta \ea\hspace{-1mm}.\hspace{-2mm}\normalsize
		\end{multline}
		
		The expression of the closed-loop maps can be obtained by multiplying (\ref{cutata}) to the left with the inverse of the square TFM on the left-hand side. 
		Doing so, we get the fact that $\mathbf{T}_{\mathbf{Q}}^{\beta w}=\mathbf{G}\tilde {\bf X}_{\bf Q} \tilde {\bf N}$ and by employing the identity 
		$\tilde {\bf X}_{\bf Q} \tilde {\bf N}=I_m-\mathbf{MY_Q}$, which can be
		deduced from \eqref{dcrelQ}, we finally obtain the fact that
		$$
		\mathbf{T}_{\mathbf{Q}}^{\beta w}=\mathbf{G}(I_m-\mathbf{MY_Q})=\mathbf{G}-\mathbf{N}\mathbf{M}^{-1}\mathbf{MY_Q}=\mathbf G-\mathbf{NY_Q}.
		$$
		
		If ${\bf G}$ is unstable, then so is $\mathbf{T}_{\mathbf{Q}}^{\beta w}$ for any stable $\mathbf{Q}$, since $\mathbf{NY_Q}$ is guaranteed to be stable, thus making the implementations in \eqref{redSLS1}-\eqref{redSLS2} unable to stabilize the feedback loop.
		\qed

		
		\bibliographystyle{IEEEtran}
		\bibliography{manuscript} 
		
	\end{document}